\documentclass[floatfix,reprint,,superscriptaddress,amsmath,amssymb,aps,prx,longbibliography,nofootinbib]{revtex4-2}
\usepackage[english]{babel} 

\usepackage[margin=1in]{geometry} 
\usepackage{qcircuit}
\usepackage{amsmath}
\usepackage{amsthm}
\usepackage{amssymb}
\usepackage{mdframed}
\usepackage{amsbsy}
\usepackage{mathrsfs}
\usepackage{braket}
\usepackage[utf8]{inputenc}
\usepackage{graphicx}
\usepackage{algorithm}
\usepackage{algpseudocode}
\usepackage{esvect}
\usepackage{dsfont}
\usepackage{cancel}

\usepackage{lineno}

\definecolor{black}{rgb}{0.86, 0.08, 0.24}

\newcommand{\footnoteI}{\footnote{\color{black}{Note that we use a different normalization convention for the Choi-Jamio\l{}kowski matrix from the authors in Ref.~\cite{nechita2018almostallquantumchannels}, giving us different normalization factors in the upper and lower bounds.}}}

\usepackage[colorlinks=false, linkcolor=black, citecolor=black, urlcolor=blue]{hyperref}
\usepackage[capitalize]{cleveref}
\usepackage[all]{hypcap}

\usepackage[usenames,dvipsnames]{xcolor}

\hypersetup{
    colorlinks=True,
    linkcolor=blue,
    citecolor=cyan,
    pdfborder={0 0 0},   
    pdfauthor={Author One, Author Two, Author Three},
    pdftitle={A simple article template},
    pdfsubject={A simple article template},
    pdfkeywords={article, template, simple},
    pdfproducer={LaTeX},
    pdfcreator={pdflatex}
}

\newcommand{\subsectionA}[1]{}

\usepackage[normalem]{ulem}
\usepackage{bm}

\newcommand{\be}{\begin{equation}}
\newcommand{\ee}{\end{equation}}

\newcommand{\Tr}{{\rm Tr\,}}

\newcommand{\e}{e}
\newcommand{\rd}{\mathrm{d}}

\newcommand{\normp}[2]{\norm{#1}_{#2}}
\newcommand{\CL}{\mathcal{L}}
\newcommand{\vrho}{\bm{ \rho}}

\newcommand{\tr}[1]{\operatorname{tr}\lb#1\rb}
\def\lb{\left(}
\def\rb{\right)}

\newcommand{\vertiii}[1]{{\left\vert\kern-0.25ex\left\vert\kern-0.25ex\left\vert #1 
    \right\vert\kern-0.25ex\right\vert\kern-0.25ex\right\vert}}

\usepackage{enumitem}

\newcommand{\cL}{\mathcal{L}}

\newtheorem{thrm}{Theorem}
\crefname{thrm}{Theorem}{Theorems}
\Crefname{thrm}{Theorem}{Theorems}

\newtheorem{cor}{Corollary}
\newtheorem{lem}{Lemma}

\newtheorem{dfn}{Definition}

\usepackage{textgreek}

\newcommand{\newreptheorem}[2]{%
  \newenvironment{rep#1}[1]{%
   \def\rep@title{#2 \ref*{##1}}%
   \begin{rep@theorem}}%
   {\end{rep@theorem}}}

\newreptheorem{lemma}{Lemma}

\newcounter{mycounter}
\usepackage{physics}
\DeclareMathOperator\erfc{erfc}

\begin{document}
\title{Efficient Quantum Gibbs Sampling with Local Circuits}

\author{Dominik Hahn}
\email{dominik.hahn@ox.ac.uk}
\affiliation{Rudolf Peierls Centre 
for Theoretical Physics, 
University of Oxford, Oxford OX1 3PU, United Kingdom}
\affiliation{IBM Quantum, IBM Thomas J. Watson Research Center, Yorktown Heights, NY 10598, USA}
\author{Ryan Sweke}
\affiliation{IBM Quantum, IBM Research – Almaden, San Jose CA, 95120, USA}
\affiliation{African Institute for Mathematical Sciences (AIMS), South Africa}
\affiliation{Department of Mathematical Sciences, Stellenbosch University, Stellenbosch 7600, South Africa}
\affiliation{National Institute for Theoretical and Computational Sciences (NITheCS), South Africa}
\author{Abhinav Deshpande}
\affiliation{IBM Quantum, IBM Research – Almaden, San Jose CA, 95120, USA}
\author{Oles Shtanko}
\email{oles.shtanko@ibm.com}
\affiliation{IBM Quantum, IBM Thomas J. Watson Research Center, Yorktown Heights, NY 10598, USA}

\begin{abstract}
The problem of simulating the thermal behavior of quantum systems remains a central open challenge 
in quantum computing. 
Unlike
well-established quantum algorithms for unitary dynamics, \emph{provably efficient} algorithms for preparing thermal states---crucial for probing equilibrium behavior---became available only recently  
with breakthrough algorithms based on the simulation of well-designed dissipative processes, a quantum-analogue to Markov chain Monte Carlo (MCMC) algorithms. We show a way to implement these algorithms avoiding expensive block encoding and relying only on dense local circuits, akin to Hamiltonian simulation. Specifically, our method leverages spatial truncation and Trotterization of exact quasilocal dissipative processes. We rigorously prove that the approximations we use have little effect on rapid mixing at high temperatures and allow convergence to the thermal state with small bounded error. Moreover, we accompany our analytical results with numerical simulations that show that this method, unlike previously thought, is within the reach of current generation of quantum hardware. These results provide the first provably efficient quantum thermalization protocol implementable on near-term quantum devices, offering a concrete path toward practical simulation of equilibrium quantum phenomena.
\end{abstract}
\maketitle

\vspace{-0.6cm}

The exponential complexity of quantum many-body systems has long posed a formidable barrier to our understanding of condensed matter, chemistry, and high-energy physics. While classical methods have achieved notable success in simulating some quantum systems, they often break down in strongly correlated regimes---precisely where the most intriguing quantum phenomena emerge, such as high-temperature superconductivity or quantum phase transitions~\cite{vojta2003quantum}. Quantum computers, originally envisioned by Feynman as simulators for these complex systems~\cite{feynman1982simulating,Nielsen2010}, are now becoming a viable platform for tackling such challenges as advances in quantum hardware accelerate rapidly \cite{kim2023evidence}.

Among the central goals in this domain is the preparation of thermal, or Gibbs, states---quantum analogues of classical equilibrium distributions that encode all information about a system at finite temperature~\cite{alhambra2023quantum}. These states are  crucial for exploring quantum statistical mechanics~\cite{gogolin2016equilibration} and serve as important computational primitives in quantum machine learning, inference, and optimization~\cite{brandao_quantum_sdp_solvers2019,kieferova2016tomography}.

Despite their importance, the efficient preparation of Gibbs states remains a central open problem in quantum computing. Classically, Markov Chain Monte Carlo (MCMC) techniques, such as the Metropolis sampling algorithm~\cite{metropolis1953equation}, 
have long been established as universal and efficient algorithms for simulating classical thermal states in many relevant cases~\cite{levin2017markov}. However, while quantum algorithms for
quantum thermal state preparation---including
phase estimation techniques~\cite{poulin2009sampling}, thermal bath emulation \cite{shtanko2021preparing}, quantum imaginary time evolution~\cite{motta2020determining},  variational methods~\cite{consiglio2023variational,deshpande2024dynamicparameterizedquantumcircuits, Ilin_2025} and
quantum Metropolis algorithms~\cite{temme2011quantum,yung2012quantum}---have 
shown substantial progress, each comes with tradeoffs in terms of scalability, provable guarantees, or compatibility with near-term devices.

A particularly promising direction involves the use of open-system dynamics inspired by 
classical MCMC, wherein a quantum system evolves under dissipative evolution that converges to a desired thermal state~\cite{davies1974markovian,rall2023thermal}. Recent breakthroughs \cite{chen2023quantum,chen2023efficient} have led to developments 
of quasilocal dissipative processes that converge to the Gibbs state with provable 
guarantees on mixing time, at least in certain regimes~\cite{rouzé2024optimalquantumalgorithmgibbs,vsmid2025polynomial,tong2024fast}. Using the direct analogy with classical MCMC, we will refer to this approach below as quantum Markov Chain Monte Carlo (or qMCMC). While earlier schemes required a continuous family of jump operators in the corresponding Lindblad equation, it has recently been shown that qMCMC can be implemented with as few as one jump operator~\cite{ding2024efficientquantumgibbssamplers,gilyen2024quantum}. A major challenge that remains is the intrinsic quasilocality of these dissipative processes. As a result, current implementations~\cite{chen2023efficient,ding2024efficientquantumgibbssamplers} cannot exploit techniques from local Hamiltonian simulation and must instead rely on block-encoding methods, which are difficult to realize on noisy quantum hardware~(but see Ref.~\cite{chen2024randomizedmethodsimulatinglindblad,brunner2025lindbladengineeringquantumgibbs,hahn2025provablyefficientquantumthermal} for alternative approaches).

In this work, we show how to address this remaining challenge. Specifically, we demonstrate that applying spatial truncation and Trotterization to qMCMC yields thermalization protocols that can be implemented directly with only \textit{local} quantum circuits. Importantly, we provide rigorous theoretical bounds showing that our proposed protocol preserves logarithmic mixing time for high temperatures. Additionally, in order to facilitate the execution of this algorithm on near-term quantum devices, we show that variational compilation can be used to compile the local processes into circuits which can be executed on currently available quantum hardware, in a way that allows one to trade off accuracy with circuit depth. As such, our method opens a new path toward thermal state preparation that is both analytically grounded and experimentally feasible on near-term quantum devices.
 
More specifically, our implementation scheme consists of three main steps. The first addresses one of the central challenges in existing qMCMC algorithms: the quasilocal nature of the dissipation. Our approach is to provide truncated versions of these dissipative processes, resulting in a strictly local Lindbladian generator amenable to circuit-based quantum simulation. 
This process can be understood as constructing a Lindbladian satisfying a \emph{local detailed-balance} condition: each jump operator satisfies detailed-balance condition with respect to a Hamiltonian truncated to the system subset. 

We show that jump operators computed on different, overlapping Hamiltonian patches collectively generate a Gibbs state of the full (untruncated) Hamiltonian.
To establish this, we extended the rapid mixing analysis from~\citet{rouze2024efficientthermalizationuniversalquantum}, originally applied to the scheme in~\citet{chen2023efficient}, to the more recent construction proposed in~\citet{ding2024efficientquantumgibbssamplers}. This extension enables a more efficient Gibbs sampler using only a few jump operators per qubit. At sufficiently high temperatures, we show that the output of the truncated dynamics converges to the true Gibbs distribution when the truncation radius scales at most logarithmically with system size. In contrast, accurate estimation of local observables requires only a truncation radius that is independent of system size. These results, formalized in \cref{thrm:trunc,thrm:mixingtime,thrm:trunc_local} below, provide the first provably efficient construction of local dissipative preparation of Gibbs states.

Next, having obtained a truncated Lindbladian with provable thermalization guarantees, we employ a randomized Trotterization strategy~\cite{chen2024randomizedmethodsimulatinglindblad} and demonstrate that its continuous-time evolution can be faithfully approximated by a circuit composed of small local quantum channels.
This construction ensures that the convergence properties of the continuous generator are retained, while the Trotter error scales quadratically with the step size.

Finally, to bridge the gap between these finite-size channels and operations available on realistic gate-based quantum devices, we 
employ a variational compilation framework.
In particular, we give compelling numerical evidence that variational compilation allows one to compile these local channels into circuits with a single qubit reset operation, and a fixed desired number of two-qubit gates, respecting hardware connectivity. Additionally, we show that this framework allows one to trade off circuit depth with accuracy, which facilitates more accurate simulations in step with yearly improvements in gate fidelity.

The proposed protocol provides a practical alternative to the amplitude‑encoding and controlled‑unitary frameworks introduced in earlier works \cite{chen2023efficient,ding2024efficientquantumgibbssamplers}. In particular, it is more readily implementable on near-term and early fault‑tolerant quantum devices, which often face stringent constraints on circuit depth and qubit connectivity.

 We validate the full protocol through numerical simulations on one- and two-dimensional spin systems. The results confirm rapid convergence to thermal states with low energy density and accurate local observables, even at modest truncation radii that include only few near neighbors. Moreover, the method exhibits some resilience to noise, underscoring its applicability to near-term quantum devices.

The remainder of the paper is organized as follows. 
In \cref{sec:background}, we introduce the theoretical framework underlying quantum thermalization via Lindbladian dynamics and review the construction of quantum Gibbs samplers based on the detailed‐balance condition. 
\Cref{subsec:Truncation} presents our main results on the truncation of quasi‐local jump operators, including formal error bounds and implications for local observables. 
In \cref{subsec:Randomizedcompilation}, we describe the randomized Trotterization scheme and derive its associated error estimates. 
\Cref{subsec:Compilation} details our variational compilation strategy for mapping the resulting local gadgets onto shallow quantum circuits. 
Finally, \cref{sec:Numerics,sec:hardware_sim} report numerical experiments on spin chains: first in an idealized setting, then incorporating realistic hardware constraints.

\section{\label{sec:background}Theoretical Background}

This section introduces the necessary background for understanding how to prepare a Gibbs state using the qMCMC method based on Lindblad time evolution. A comprehensive list of all definitions  and notations for the entire paper is available in \cref{sec:notations}.

We consider a finite $D$-dimensional square lattice $\Lambda$ consisting of $n$ sites such that each site $a \in \Lambda$ is associated with a qubit. Let the total Hilbert space of the lattice be $\mathsf H_\Lambda$, the corresponding algebra of operators over this space be denoted by  $\mathsf M_\Lambda$, and the set of Hermitian operators denoted by $\mathsf S_\Lambda$. Then, given $H \in \mathsf S_\Lambda$ and a non-negative inverse temperature ${\beta \geq 0}$, our goal is to prepare the Gibbs state defined as \begin{align}\label{Eq:Gibbs state}
    \rho_\beta:=\frac{e^{-\beta H}}{\tr e^{-\beta H}}.
\end{align}
In the remainder of this paper, we focus on local Hamiltonians of the form 
\begin{equation}\label{eq:ham}
H = \sum_{X\subset\Lambda} H_X,
\end{equation}
where each term $H_X$ is strictly supported on a subset $X\subseteq B_{a_\mu}(r_\mu)\subset\Lambda$ involving a constant number of qubits inside the ball $B_{a_\mu}(r_\mu)$ of constant radius $r_\mu$ around qubit $a_\mu\in\Lambda$.

In the 1970s, Davies proposed that, in the weak-coupling limit, the dynamics of a system coupled to a large thermal bath at inverse temperature $\beta$ is described by a Markovian master equation that has the Gibbs state for its Hamiltonian $H$ as its unique fixed point~\cite{davies1974markovian,davies1976markovian}. The Davies process, when expressed in the Hamiltonian's eigenbasis, typically induces transitions between pairs of eigenstates, resulting in nonlocal dynamics that often involve all qubits in the system.

Recent works~\cite{chen2023quantum,chen2023efficient,ding2024efficientquantumgibbssamplers} have shown that the Davies process arises as a special case within a broader family of thermalizing Markovian dynamics that respects the detailed balance condition (see Eq.~\eqref{eq:KMS}). In particular, the foundational work of \citet{chen2023efficient} demonstrated that it is possible to construct a process that converges exactly to the Gibbs state even when the transitions are modified to have finite energy resolution, rather than mapping directly from eigenstate to eigenstate. Building on this, the work by 
\citet{rouze2024efficientthermalizationuniversalquantum} showed that at sufficiently high temperatures, such algorithms exhibit rapid mixing, achieving convergence to the Gibbs state in logarithmic time in the number of qubits. In this work, we adopt a more recent construction proposed by \citet{ding2024efficientquantumgibbssamplers}, which, in contrast to earlier approaches, requires only a finite number of jump operators per qubit. 

Specifically we consider dynamics that, acting on a density matrix $\rho$, 
can be described by a Gorini–Kossakowski–Sudarshan–Lindblad equation
\begin{align}\label{eq:LindbladLin}
    \frac{d}{dt}\rho = \mathcal{L}^\beta(\rho) = \sum_{a \in \Lambda} \sum_\alpha \mathcal{L}^\beta_{a,\alpha}(\rho),
\end{align}
where each index $\alpha$ takes discrete values. Each local Lindbladian term\footnote{While terminology may vary across the literature, we use the term \textit{Lindblad operators} or \textit{Lindbladians} to refer to the full operators $\mathcal{L}^\beta$ and $\mathcal{L}^\beta_{a,\alpha}$, including both their dissipative and coherent components. This contrasts with definitions that reserve the term only for the dissipative part.} $\mathcal{L}^\beta_{a,\alpha}$ takes the form
\begin{align}\label{eq:lindbladian_main}
\mathcal{L}^\beta_{a,\alpha}(\rho) = 
&-i[G^\beta_{a,\alpha}, \rho] + \nonumber\\
&L_{a,\alpha}^{\beta} \rho L^{\beta\dagger}_{a,\alpha} - \frac{1}{2}\{L_{a,\alpha}^{\beta\dagger} L^\beta_{a,\alpha}, \rho\},
\end{align}
where $G^\beta_{a,\alpha} \in \mathsf S_\Lambda$ are Hermitian operators that we define later, and the jump operators $L^\beta_{a,\alpha} \in \mathsf M_\Lambda$ are given by
\begin{align}\label{eq:jump}
    L^\beta_{a,\alpha} = \int_{-\infty}^{\infty} \mathrm{d}t \, f(t) e^{iHt} A^{a,\alpha} e^{-iHt}.
\end{align}

Here, we choose $A^{a,\alpha}$ to be single-qubit Hermitian operators supported on the qubit at lattice site ${a \in \Lambda}$, as per Eq.~\eqref{jump_operator_choice}. Also,
$f(t)$ is a filter function given by
\begin{align}\label{eq:envelope}
    f(t) = \frac{1}{2\pi} \int_{-\infty}^\infty \mathrm{d}\nu\, q(\nu) e^{-\beta \nu/4} e^{i t \nu},
\end{align}
where the envelope function $q(\nu)$ is symmetric, satisfying $q(-\nu) = q^*(\nu)$. It can be chosen such that the resulting filter function effectively suppresses contributions to the integral in Eq.~\eqref{eq:jump} at large times $t$. For local Hamiltonians $H$ that satisfy a Lieb-Robinson bound, this implies that the jump operators $L^\beta_{a,\alpha}$ are quasilocal; in other words, they are predominantly supported on qubits near site $a$.

To define the coherent term, consider the spectral decomposition of the Hamiltonian $H = \sum_i \lambda_i P_i$, where $\lambda_i\in\mathbb R$ are eigenvalues and $P_i\in\mathsf S_\Lambda$ are the corresponding spectral projectors. The coherent term $G^\beta_{a,\alpha}\in\mathsf S_\Lambda$ is explicitly given by
\begin{align}\label{eq:coherentLin}
    G^\beta_{a,\alpha} = -\frac{i}{2}  \sum_{\nu \in \Omega(H)} \tanh\left(-\frac{\beta \nu}{4}\right) (L_{a,\alpha}^{\beta\dagger} L^\beta_{a,\alpha})_\nu,
\end{align}
where $\Omega(H) = \{\lambda_i - \lambda_j \mid i,j\}$ denotes the set of Bohr frequencies of $H$, and for any operator $A\in\mathsf M_\Lambda$, the component $A_\nu\in\mathsf M_\Lambda$ is defined as
\begin{align}\label{eq:proj}
    A_\nu := \sum_{i,j:\lambda_i - \lambda_j = \nu} P_i A P_j,
\end{align}
where the sum is taken over all pairs of energy levels separated by frequency $\nu$.\footnote{Note that \citet{ding2024efficientquantumgibbssamplers} introduced an additional regularization for large frequencies $\nu$ to ensure a well-defined time representation in the thermodynamic limit. We avoid the necessity of regularization by choosing either rapidly decaying filter functions or restricting to small subsystems, as becomes clear below.}

It can be demonstrated~\cite{ding2024efficientquantumgibbssamplers} that the chosen filter function and coherent term $G^\beta_{a,\alpha}$ guarantee that the Lindbladian $\mathcal{L}^\beta$ satisfies the Kubo-Martin-Schwinger (KMS) detailed balance condition
\begin{align}\label{eq:KMS}
\mathcal{L}^{\beta \dagger}(\cdot ) = \rho_\beta^{-\frac{1}{2}}\, \mathcal{L}^\beta\Bigl(\rho_\beta^{\frac{1}{2}}\, \cdot \, \rho_\beta^{\frac{1}{2}}\Bigr)\, \rho_\beta^{-\frac{1}{2}},
\end{align}
where the adjoint superoperator $\mathcal{L}^{\beta \dagger}$ is defined with respect to the inner product
$
\langle A, B \rangle = 2^{-n}\operatorname{Tr}\!\left(A^\dagger B\right)
$
via the relation
$
\langle A, \mathcal{L}^\beta(B) \rangle = \langle \mathcal{L}^{\beta \dagger}(A), B \rangle.$ The KMS condition, expressed in Eq.~\eqref{eq:KMS}, can be viewed as an extension of the detailed balance condition known from classical Markov processes \cite{levin2017markov} to the quantum setting~\cite{chen2023efficient}. Directly from Eq.~\eqref{eq:KMS}, it follows that the Gibbs state $\rho_\beta$ is a steady state of the Lindbladian~\cite{chen2023quantum}, i.e.,
\begin{align}
    \mathcal{L}^\beta(\rho_\beta)=0.
\end{align}

For concreteness, unless otherwise stated, we choose the envelope function $q(\nu)$ in Eq.~\eqref{eq:envelope} to be a Gaussian,
\begin{align}\label{eq:Gauss}
    q(\nu)=\exp\left(-\frac{(\beta\nu)^2}{8}\right).
\end{align}
With this specific choice of $q(\nu)$, it is possible to give an explicit expression for $f(t)$ and $G^\beta_{a,\alpha}$ as integrals in the time domain. With the Gaussian choice in Eq.~\eqref{eq:Gauss}, $f(t)$ is given by (see \cref{sec:time})
\begin{align}\label{eq:timejump}
    f(t)=\sqrt{\frac{2}{\pi \beta^2}}\exp\left(\frac{(\beta-4 i t)^2}{8\beta^2}\right)
\end{align}
and
\begin{align}\label{eq:coherenttime}
    G^\beta_{a,\alpha}=&
		\int_{-\infty}^{\infty}\!\rd tg_1(t)\e^{-iH t}\times \\
        \left(\int_{-\infty}^{\infty} \right. & \left.\rd t'g_2(t')\e^{iH t'}A^{{a,\alpha}\dagger}\e^{-2iH t'} A^{a,\alpha}\e^{iH t'} \right)\e^{iH t}
\end{align}
with
\begin{align}\label{eq:g1time}
     g_1(t)=\left[-\frac{1}{\pi \beta \cosh(\frac{2\pi t}{\beta})}\right]*_t\left[\frac{\sqrt{2}}{\beta}e^{\frac{1}{4}-\frac{4t^2}{\beta^2}}\sin\left(\frac{2 t}{\beta}\right)\right],
\end{align}
where $f(t)*_tg(t):=\int_{-\infty}^\infty \rd s f(s)g(t-s)$ denotes the convolution, and 
 \begin{align}
\label{eq:g2time}
 g_2(t) =\frac{2\sqrt{2}}{\beta} \exp\biggl(\frac{(\beta-4 it)^2}{4 \beta^2}\biggl).
\end{align}
Below, we consider operators in the form of single-qubit Pauli operators, specifically 
\be\label{jump_operator_choice}
A^{a,1}=X_a,\quad  A^{a,2}=Y_a, \quad A^{a,3}=Z_a.
\ee
This specific choice of jump operators is motivated by theoretical convenience: in the limit $\beta \to 0$, the resulting Lindbladian describes depolarizing noise (see \cref{subsec:beta0}). This simplification streamlines the ensuing theoretical analysis for deriving the error and performance‐time bounds. Nevertheless, our framework can, in principle, accommodate a much broader class of initial operators—so long as they remain sufficiently generic (for example, non-commuting with the Hamiltonian). Also, both the compilation strategy described below and our numerical simulations methods permit arbitrary operator choices. By selecting these operators strategically, one can optimize algorithmic runtime or minimize circuit depth, tailoring them to a particular Hamiltonian or compilation scheme. A comprehensive exploration of these optimizations, however, lies beyond the scope of the present work.

The performance of quantum Gibbs samplers is characterized by their mixing time, defined as follows (e.g., see Proposition E.4 in \cite{chen2023efficient}):

\begin{dfn}\label{dfn:mixing_time}
For a given Lindbladian~$\mathcal{L}$, the mixing time is the smallest time $t_\mathrm{mix}$ satisfying
\begin{align}\label{eq:mixing_time_def}
    \forall \rho,\rho'\in \mathsf D_\Lambda:\quad \|e^{\mathcal{L}t_{\mathrm{mix}}}(\rho-\rho')\|_1\leq \frac{1}{2}\|\rho-\rho'\|_1.
\end{align}
\end{dfn}

The mixing time provides a natural way to define the runtime of the algorithm. Let $\rho$ be the initial state and $\rho'$ the fixed point of the channel $\mathcal{L}$. It follows by the definition of mixing time that the time $t$ required for the system to reach $\rho'$ within a trace distance of $\epsilon > 0$ is bounded by $t \leq t_{\mathrm{mix}} \log_2(1/\epsilon)$. Therefore, estimating the mixing time is a central objective in our analysis.

Previously, it was demonstrated~\cite{rouze2024efficientthermalizationuniversalquantum} that the quantum Gibbs sampler introduced in Ref.~\cite{chen2023efficient} exhibits a mixing time scaling linearly with the number of qubits $n$ for temperatures below a threshold $\beta < \beta^*$, where $\beta^*$ does not depend on system size. This result was obtained by establishing bounds on the Lindbladian's spectral gap, which, due to the KMS detailed balance condition in Eq.~\eqref{eq:KMS}, provides an upper bound for the mixing time.
This result was improved in Ref.~\cite{rouzé2024optimalquantumalgorithmgibbs}, which showed a logarithmic mixing time for high enough temperature, by bounding an oscillator norm.
Using analogous techniques from~\citet{rouzé2024optimalquantumalgorithmgibbs}, we demonstrate that the Lindbladian described by~\citet{ding2024efficientquantumgibbssamplers}, namely Eq.~\eqref{eq:LindbladLin}, also achieves convergence to a Gibbs state with mixing time scaling logarithmically with the system size, and we extend this result to truncated dynamics that will be introduced later. 

While we do not cover this in our work, one can also consider using similar methods that have shown that quantum Gibbs samplers effectively prepare thermal states in the Fermi-Hubbard model at arbitrary temperatures and weak interactions~\cite{vsmid2025polynomial,tong2024fast}. 

\section{\label{sec:Trotterization} Analytical Results}

In this section, we outline the steps required for a hardware-efficient simulation of the Lindbladian described by Eq.~\eqref{eq:LindbladLin}. The implementation procedure consists of three main stages: (a) truncating the jump operators and the coherent term, (b) applying Trotterization to the resulting truncated Lindbladian, and (c) compiling the elementary quantum channels that implement Trotterized evolution for each Lindbladian component. 

For the truncation step, we present formal results that establish rigorous error bounds on the truncation error. Specifically, we demonstrate that for sufficiently high temperatures, which are independent of the system size, the truncation error in evaluating both local observables and the trace distance can be controlled by increasing the truncation radius, as outlined in Corollaries~\ref{cor:trunc} and~\ref{col:loc_obs} below. In particular, Gibbs sampling generally requires the truncation radius to scale logarithmically with the number of qubits, whereas estimating the expectation value of local observables requires only a constant radius.

Subsequently, we introduce a simulation protocol for the truncated Lindbladian based on the randomized compilation approach for open system dynamics~\cite{chen2024randomizedmethodsimulatinglindblad}. This protocol allows us to approximate the truncated dynamics using a circuit consisting of local quantum channels, each of which implements the short time evolution of a truncated jump operator. 
The randomized approach reduces the depth of the circuit and makes it independent of the number of jump operators, at the cost of increased sampling overhead. The discussion of compilation of the circuit is deferred to \cref{subsec:Compilation}.

\subsection{Truncation of quasilocal Lindbladians}\label{subsec:Truncation}
\begin{figure*}[t!]
	\centering
	\includegraphics[width=0.95\textwidth]{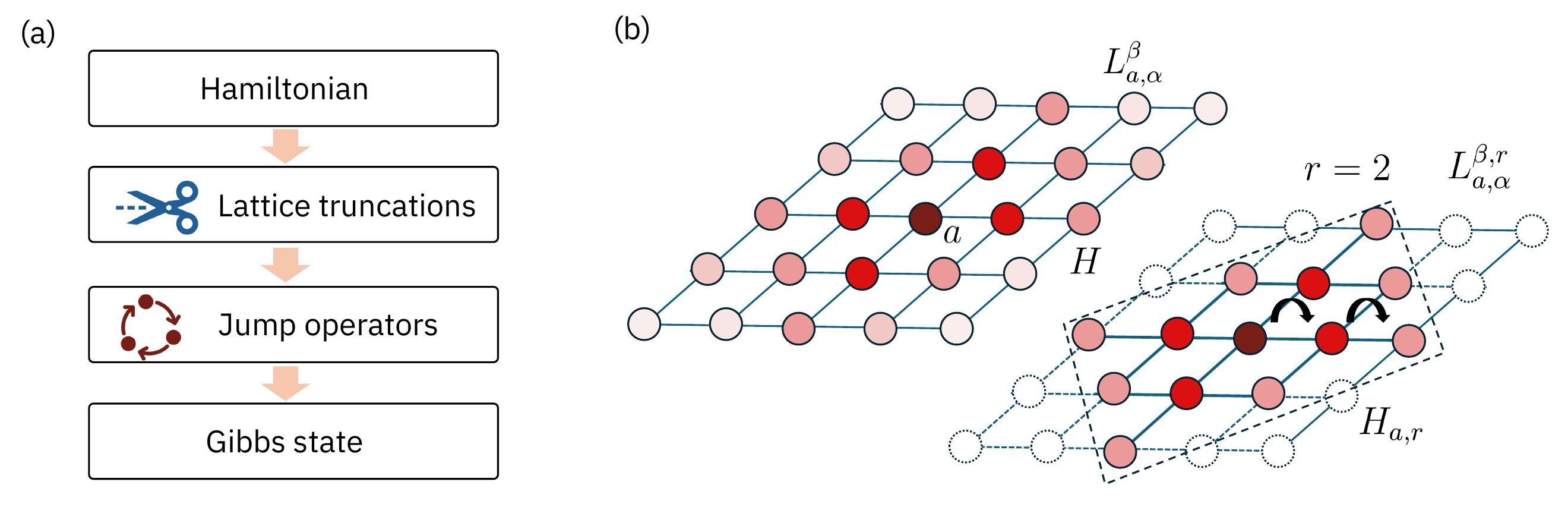}
    \caption{\textbf{Sketch of truncation-based dissipative Gibbs state preparation.}  
(a) Starting from the target Hamiltonian \(H\), we construct a family of local, truncated-lattice Hamiltonians \(H_{a,r}\). From these truncated models, we derive the corresponding local Lindblad jump operators, which define global dissipative evolution that drives the system toward the desired Gibbs state.
(b) Truncation schematic. \emph{Left:} Jump operator derived from the full-system dynamics on \(\Lambda\), chosen to be a square lattice; only a portion of the lattice is shown. The action of the jump operator decays exponentially with the graph (Manhattan) distance from the central qubit \(a \in \Lambda\). \emph{Right:} Jump operator constructed by restricting the dynamics to qubits within a graph radius \(r\) (here, \(r = 2\)), yielding strictly local support.
 }
	\label{fig:truncation}
\end{figure*}
In this section, we analyze the error arising from replacing quasilocal operators with their locally truncated versions. Consider a Hamiltonian introduced in Eq.~\eqref{eq:ham}.
Let us introduce a fixed truncation radius $r \geq 0$ and 
define the locally truncated Hamiltonian $H_{a,r}$ as the sum of all local terms whose supports are contained within the ball $B_a(r)$ of radius $r$ centered around qubit $a$,
\begin{align}\label{eq: trunc_ham}
H_{a,r} := \sum_{X \subseteq B_a(r)} H_X.
\end{align}
Using this construction, we define the truncated Lindbladian as
\be\label{eq:trunc_lindbl}
\mathcal{L}^{\beta,r}=\sum_{a\in\Lambda}\sum_\alpha\mathcal{L}_{a,\alpha}^{\beta,r}
\ee
where
\begin{align}\label{eq:TruncLindblad}
\mathcal L^{\beta,r}_{a,\alpha}(\rho) = &-i[G^{\beta,r}_{a,\alpha},\rho]+\nonumber\\
&L^{\beta,r}_{a,\alpha} \rho L^{\beta,r\dagger}_{a,\alpha}-\frac{1}{2}\{L^{\beta,r\dagger}_{a,\alpha} L^{\beta,r}_{a,\alpha},\rho\},
\end{align}
expressed through the truncated coherent terms and jump operators
\begin{align}\label{eq:trunc_LG}
    & L_{a,\alpha}^{\beta,r}=\sum_{\nu \in \Omega(H_{a,r})}  q(\nu) (A^{a,\alpha})_{\nu[a,r]}, \nonumber\\
    & G_{a,\alpha}^{\beta,r} = -\frac{i}{2}  \sum_{\nu \in \Omega(H_{a,r})} \tanh\left(-\frac{\beta \nu}{4}\right) (L^{\beta,r\dagger}_{a,\alpha} L^{\beta,r}_{a,\alpha})_{\nu[a,r]}.
\end{align}
Here, the projector $A_{\nu[a,r]}$ is defined similarly to Eq.~\eqref{eq:proj}, but using the projectors and eigenstates of the truncated Hamiltonian $H_{a,r}$. This truncation procedure is summarized in Fig.~\ref{fig:truncation}.
Consequently, the resulting truncated jump operators have compact support on $B_a(r)$.

The truncation introduces a controllable error in the resulting fixed point of the Lindbladian evolution, as established in the following theorem.

\begin{thrm}[Truncation error] \label{thrm:trunc}
The truncated Lindbladian $\mathcal{L}^{\beta,r}$ in Eq.~\eqref{eq:trunc_lindbl}, 
with the choice of complete set of generators in Eq.~\eqref{jump_operator_choice},
has a unique fixed point $\rho_{\beta,r}\in \mathsf D_\Lambda$, i.e., 
$$
\mathcal{L}^{\beta,r}(\rho_{\beta,r}) = 0.
$$
Moreover, there exist constants $\beta^*>0$ and $J>0$ such that, for any $\beta<\beta^*$
\begin{equation} \label{eq:distance_gibbs}
\|\rho_{\beta,r}- \rho_\beta\|_1 = O\left((\beta J)^{\frac{r}{2}} n \log n \right).
\end{equation}
\end{thrm}
The proof of Theorem~\ref{thrm:trunc} is provided in \cref{sec:tracedistanceerror}. 
To establish this result, we invoke Lemma~2 from Ref.~\cite{chen2023efficient}, which provides an upper bound on the distance between the fixed points of two Lindbladians, which is at most four times of the product between: (a) mixing time $t_{\mathrm{mix}}$ of one of the Lindbladians and (b) the Schatten-1 induced norm $\|\mathcal{L}^\beta - \mathcal{L}^{\beta,r}\|_{1\rightarrow 1}$, which 
remains bounded and controllable due to the (quasi)-locality of all terms appearing in each of the Lindbladians.

Consequently, to derive the bound on the mixing time, we extend the proof of rapid mixing from \citet{rouzé2024optimalquantumalgorithmgibbs} devised for the process in the work by \citet{chen2023efficient} to the algorithms we analyze, specifically the quantum Gibbs sampler in Eq.~\eqref{eq:LindbladLin} and the truncated Lindbladian in Eq.~\eqref{eq:trunc_lindbl}. For the latter, we state the following result:
\begin{thrm}[informal]\label{thrm:mixingtime} Under conditions of Theorem~\ref{thrm:trunc}, the Lindbladians $\cL^\beta$ and $\mathcal L^{\beta, r}$, with $r\geq0$ have mixing time
    \be
t_{\rm mix}=O(\log n ).
    \ee
\end{thrm}

This proof builds on showing that the Lindbladian, similar to the one introduced in \citet{chen2023efficient}, for all $ O\in \mathsf M_\Lambda$ satisfies the inequality \cite{rouzé2024optimalquantumalgorithmgibbs} 
 
 \begin{align}\label{gradientest}
\vertiii{e^{t\mathcal{L}^{\beta \dagger}}(O)}\le e^{-(1-\kappa)t}\vertiii{O}\,
\end{align}
for some $\kappa>0$ with the oscillator norm defined as
\begin{align*}
\vertiii{O}:=\sum_{a\in\Lambda}\|O-\tfrac{1}{2}I_a\otimes \operatorname{tr}_a(O)\|_\infty\,,
\end{align*}
where $\operatorname{tr}_a(\cdot)$ is partial trace over qubit $a$.
From that statement, it follows~[cf. \cref{sec:tracedistanceerror}]
\begin{align}
\|e^{t\mathcal{L}^\beta}(\rho)-\rho_\beta\|_1\le 2n \|\rho-\rho_\beta\|_1 e^{-(1-\kappa)t}\,.
\end{align}
It turns out that $\kappa$ can be bounded by exploiting the quasilocality of the Lindbladian, provided the temperature is sufficiently high, i.e., $\beta < \beta^*$. Crucially, since the proof relies on the locality of the jump operators, the bound on the mixing time for the original Lindbladian $\mathcal{L}^\beta$ is always greater than or equal to that of the truncated Lindbladian $\mathcal{L}^{\beta,r}$. In particular, truncation does not worsen the mixing time bounds at high enough temperatures.

\begin{figure*}[t!]
	\centering
	\includegraphics[width=0.9\textwidth]{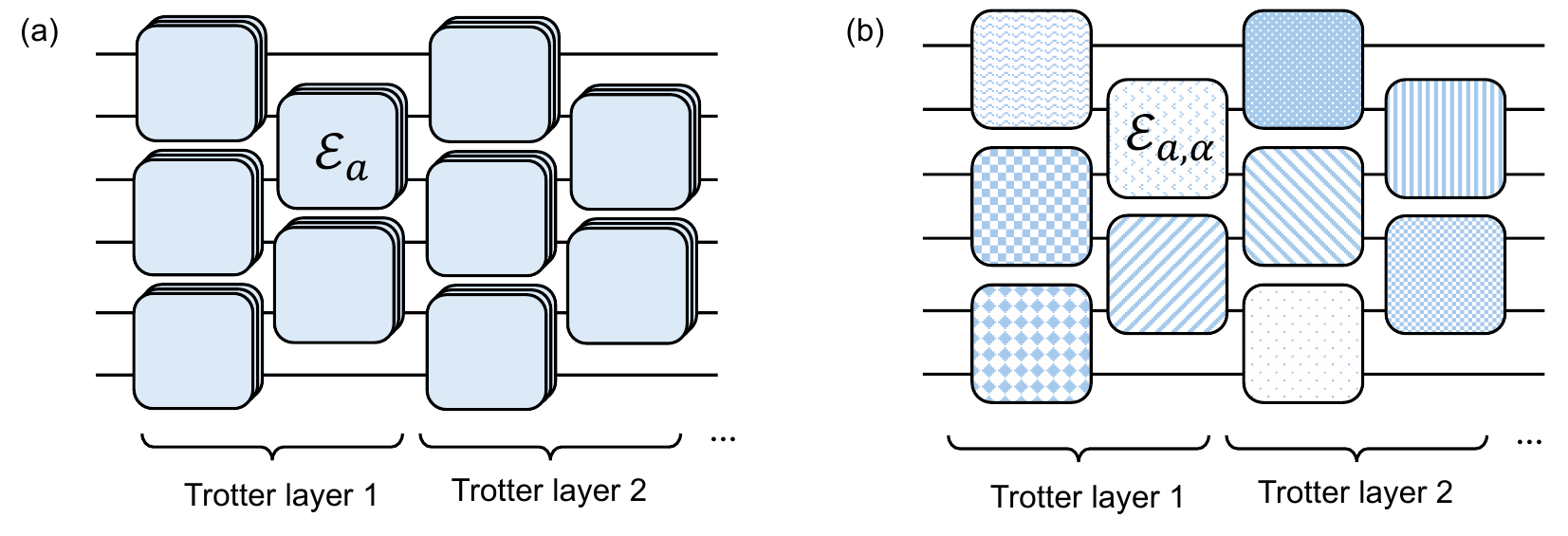}
    \caption{\textbf{Standard~vs.~randomized~Trotterization.} (a) Standard Trotterization approximates continuous-time evolution by a sequence of discrete operations $\mathcal{E}_a := \exp\left(\sum_{\alpha} \mathcal{L}^{\beta,r}_{a,\alpha} \tau\right)$, where $\tau$ is the Trotter step size. Each operation is visualized as a stack of gates representing the sum in the exponent.
    (b) To reduce circuit depth, we employ a randomized compilation strategy: at each Trotter step and for each site \(a\), we randomly select a single term \(\exp\bigl(\mathcal{L}_{a,\alpha}^{\beta,r}\,\tau\bigr)\) from \cref{eq:channelrandomvariable} by drawing \(\alpha\) uniformly from the set of possible values.
These randomized operations correspond to shallower circuits derived from the exponentiation in Eq.~\eqref{eq:average}. Averaging over many such randomized circuit realizations recovers the full dynamics.}
	\label{fig:Trotterization}
\end{figure*}

As a direct corollary of the established results so far, Theorem~\ref{thrm:trunc} provides an estimate for the required truncation radius $r$. In particular, it asserts that:
\begin{cor}\label{cor:trunc}
For any fixed $\beta < \beta^*$, the truncated Lindbladian $\mathcal{L}^{\beta,r}$ has a fixed point $\rho_{\beta,r}$ that is $\epsilon$-close to the Gibbs state $\rho_\beta$, provided that
\begin{equation}\label{eq:radius_trunc}
    r = \Omega\left(\log\left[\frac{n}{\epsilon}\right]\right).
\end{equation}
\end{cor}

The linear scaling with the number of qubits in Eq.~\eqref{eq:distance_gibbs} arises from the general subadditivity of the trace distance \cite{Nielsen2010} and is unlikely to be improved if the goal is to sample from the exact Gibbs distribution.
However, in many physical applications, the primary objective is to compute the expectation values of local observables. For this task, the requirements on the truncation radius are more relaxed as one can use restrictions such as the Lieb-Robinson bound. The result is demonstrated by the following theorem.

\begin{thrm}[Truncation error for local observables] \label{thrm:trunc_local}
Under the assumptions of Theorem 1, consider $\beta<\beta^*$ and let $O_X\in\mathsf S_X$ be supported on $X\subseteq\Lambda$. Then there exist constants $\gamma,\eta>0$, depending only on the model parameters and $\beta$, such that $\rho_{\beta,r}$ satisfies
\begin{equation}\label{eq:loc_obs_distance}
\begin{split}
	\bigl| \tr(\rho_{\beta, r} O_X) - &\tr(\rho_\beta O_X) \bigr| \\
    &\leq  \norm{O_X}c(|X|)\Bigl[e^{-\gamma r} +ne^{-\eta d_X}\Bigr]
\end{split}
\end{equation}
where $c(|X|)=O\bigl(\mathrm{poly}(|X|)\bigr)$, $n \equiv |\Lambda|$ is the number of qubits, and
\[
d_{X} \;=\; \min_{\,a\in X\,}\ell\bigl(a,\partial\Lambda\bigr)
\]
is the minimal graph‐distance from any $a\in X$ to the boundary $\partial\Lambda$.
\end{thrm}
For a fixed region $X$, the distance $d_X$, which is distance between $X$ and to the complement $\Lambda^c$ of the lattice $\Lambda$ and is thus the distance to the boundaries, increases with system size, which implies that the term $|\Lambda| \nu_\eta^{-1}(d_X)$ vanishes as the system becomes large. This term can therefore be interpreted as a finite-size correction~\cite{cubitt2015stability}.
The proof of Theorem~\ref{thrm:trunc_local} is detailed in \cref{section:Gap stability}.  This proof relies on a theorem concerning the stability under local perturbations of Lindbladians with logarithmic scaling of the mixing time, as derived in Ref.~\cite{cubitt2015stability}. In fact, this stability theorem allows us to prove an even stronger statement. Specifically, consider the time evolution generated by $\cL^\beta$, $\exp(t\mathcal{L}^\beta)$, and the time evolution $\exp(t\mathcal{L}^{\beta,r})$. We show that for any observable $O_X$ supported on $X$, the distance between $\exp(t\mathcal{L}^{\beta \dagger })[O_X]$ and  $\exp(t\mathcal{L}^{\beta,r\dagger})[O_X]$ remains small for \textit{all} times $t$. The result then follows from considering the limit $t\rightarrow \infty$. We emphasize that the support $X$ does not necessarily have to be connected, thus this theorem also applies to correlation functions. Also, note that the bound of $c(|X|)$ by a polynomial is useful when $|X|$ is small, but it becomes useless for global observables.

A direct corollary of Theorem~\ref{thrm:trunc_local}  is: 
\begin{cor}\label{col:loc_obs}
For any $\beta < \beta^*$, the truncated Lindbladian $\mathcal{L}^{\beta,r}$ admits a fixed point 
for which the expectation values of observables whose support size is $O(1)$ and sufficiently distant from the lattice boundary are
$\epsilon$-close to those of the Gibbs state $\rho_\beta$, provided that
\begin{equation}
    r = \Omega(\log(1/\epsilon)).
\end{equation}
\end{cor}
Notably, this requirement is independent of the system size and applies to observables located far from the lattice boundaries, where the last term in Eq.~\eqref{eq:loc_obs_distance} can be neglected; that is, when $d_X = \Omega(\log n)$.

In addition to the analytical results above, we also present in \cref{sec:Numerics,sec:Further models} numerical experiments aimed at evaluating the error in expectation values of local observables in certain one-dimensional models.
Our findings indicate that the error remains modest even at low temperatures (large inverse temperatures $\beta$), a regime where the aforementioned theoretical bounds may not strictly apply. Consequently, these numerical results suggest that the practical validity of truncation methods extends beyond the analytically predicted limits, at least for one-dimensional systems.

\subsection{Trotterization of truncated dynamics}\label{subsec:Randomizedcompilation}

After deriving the truncated Lindblad operator and demonstrating its efficiency in preparing Gibbs states, the subsequent step toward practical implementation is to discretize the evolution using Trotterization. Consider an evolution over total time $t = \tau M$, where $M$ is a large integer. Then, we can partition the continuous evolution into discrete time steps of length $\tau$ and approximate the dynamics by
\begin{align}
\exp(\mathcal L^{\beta,r}t)\to \left[\prod_{a \in \Lambda} \exp\left(\sum_{\alpha} \mathcal{L}^{\beta,r}_{a,\alpha}\tau\right)\right]^M,
\label{eq:trotterized}
\end{align}
where $\prod_{a \in \Lambda}$ denotes a composition of channels applied in an arbitrary but fixed order, see Fig.~\ref{fig:Trotterization}(a). This method generalizes the Trotterization approach commonly used in Hamiltonian simulations and is known to introduce an error scaling quadratically with time $t$~\cite{Kliesch2011Dissipative,Han2021Experimental}, as we will see later in this Section.

\begin{figure*}[t!]
	\centering
	\includegraphics[width=0.8\textwidth]{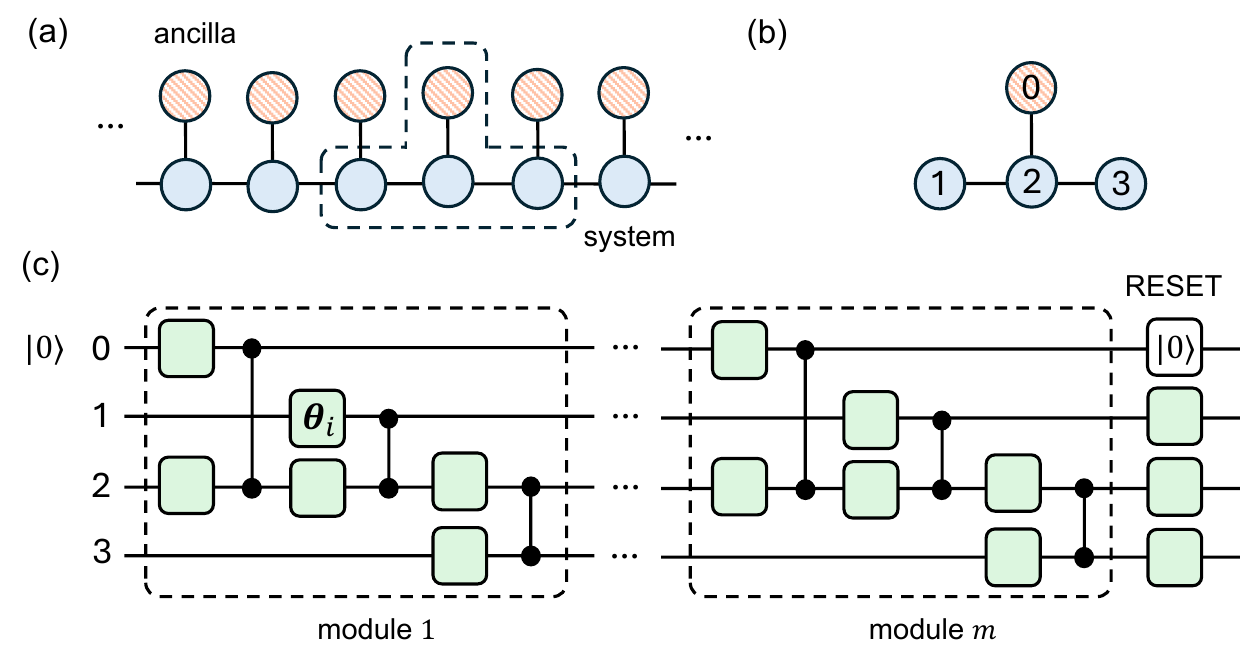}
    \caption{
    \textbf{Gadgets for one-dimensional Hamiltonians.} (a) The ``ladder'' device architecture, shown as an example of a local layout for simulating thermal states of one-dimensional Hamiltonians. The lower row of qubits (solid blue) represents the system qubits, while the upper row (dashed red) represents the ancilla. (b) The gadget connectivity for a truncation radius of $r = 1$, which requires implementing a four-qubit unitary. (c) A template circuit example, consisting of $m$ individually optimized modules followed by single-qubit gates. The ancilla is initialized in the state $\ket{0}$ before the gadget is applied and is reset afterwards to ensure correct initialization for the next gadget. Each module consists of three controlled-$Z$ gates arranged in an alternating pattern, interleaved with single-qubit gates parametrized as in Eq.~\eqref{eq:singlequbit}.
}
	\label{fig:Ansatz}
\end{figure*}

To further reduce circuit depth, we introduce a randomized compilation method tailored specifically for Lindblad dynamics at each discrete time step~\cite{chen2024randomizedmethodsimulatinglindblad}. We define a simulation trajectory
\begin{align}\label{eq:channelrandomvariable}
    \mathcal{E}^{\bm \alpha}_{t,\tau} := \prod_{i=1}^M \prod_{a \in \Lambda} \exp\left( \mathcal{L}_{a,\alpha[a]}^{\beta,r} \tau \right),
\end{align}  
where $\alpha[a]$ represents a single choice of operator over site $a$,  $\bm \alpha = \{\alpha[a] \mid a\in \Lambda\}$. Next, we consider the index $\alpha[a]$ to be selected independently and uniformly at random, thus inducing a uniform probability distribution $\alpha[a]\sim \Sigma$. Then we obtain that for sufficiently large $M$ (see Theorem~\ref{thrm:trot_thrm} below):
\begin{align}
    \exp\left(\mathcal{L}^{\beta,r} t\right) \approx \mathbb{E}_{\bm{\alpha} \sim \Sigma^{\otimes n}} \left[\mathcal{E}^{\bm \alpha}_{t,\tau}\right],
    \label{eq:average}
\end{align}
where the r.h.s. represents the average of the randomized, discretized evolution with time step $\tau$. Practically, this evolution is realized by randomly selecting Lindbladian terms indexed by different values of $\alpha$ at each circuit layer and averaging over the resulting quantum trajectories, as illustrated in Fig.~\ref{fig:Trotterization}(b).

The following results shows that for sufficiently large $M$, the output of these expectation values converges to the expectation value of the target evolution.
\begin{thrm}[Theorem 4 from Ref.~\cite{chen2024randomizedmethodsimulatinglindblad}]\label{thrm:trot_thrm}
The error of the average channel $\mathcal{E}_{t,M}$ in Eq.~\eqref{eq:channelrandomvariable} can be bounded by
\begin{align}\label{eq:trotter_error}
   \norm{\mathbb{E}_{\bm{\alpha} \sim \Sigma^{\otimes n}} \left[\mathcal{E}^{\bm \alpha}_{t,\tau}\right]-e^{t \mathcal{L}^\beta}}_\diamond =O\left(\frac{n^2 t^2}{M}\right).
\end{align}
\end{thrm}
The proof of this theorem is in Ref.~\cite{chen2024randomizedmethodsimulatinglindblad}.

\section{Circuit Compilation}\label{subsec:Compilation}
The final step required for an implementation of the Trotterized dynamics that is suitable for near-term quantum hardware is a compilation of the local quantum channels into circuit-level gate operations. In principle, this can be done efficiently using dilation methods, followed by compilation into the desired gate set. According to standard results in the literature \cite{Nielsen2010}, for a fixed $r$, this compilation requires a number of gates that scales at most polylogarithmically with the inverse of the additive error.

Here we show that a single ancilla qubit is in fact sufficient, and that variational compilation of the resulting dilated unitary provides a practically feasible method for yielding sufficiently short circuits for near-term implementation.  More specifically, we provide a compilation procedure for each channel \(\exp(\tau \mathcal{L}_{a,\alpha}^{\beta,r})\), which represents the randomized evolution described in Eq.~\eqref{eq:channelrandomvariable}.
This can be done by first defining the Hermitian operator
\begin{align}\label{eq:V_operator}
    O^{a,\alpha}_r := \begin{pmatrix}
    \sqrt{\tau}\,G_r^{a,\alpha} & L_r^{a,\alpha\dagger} \\
    L_r^{a,\alpha} & \sqrt{\tau}\,G_r^{a,\alpha}
    \end{pmatrix},
\end{align}
where \(L_r^{a,\alpha}\) and \(G_r^{a,\alpha}\) are the jump operators and associated coherent terms, respectively, as was previously defined in Eq.~\eqref{eq:trunc_LG}.

Then we can approximate (see Lemma~\ref{lem:Quantumgate})
\be\label{eq:c_channel}
\exp\left( \mathcal{L}_{a,\alpha}^{\beta,r} \tau \right) \approx \mathcal C^{a,\alpha}_{r,\tau}(\rho),
\ee
where we introduced the quantum channel
\be
\mathcal C^{a,\alpha}_{r,\tau}(\rho):= \text{Tr}_{\text{anc}}\left(U^{a,\alpha}_r(\rho \otimes |0\rangle\langle 0|_{\text{anc}})U^{a,\alpha\dagger}_r\right),
\ee
where \(\text{Tr}_{\text{anc}}(\cdot)\) denotes the partial trace over the ancilla qubit and the unitary \(U^{a,\alpha}_r\) describes the evolution over a time interval \(\sqrt{\tau}\) as
\begin{equation}\label{eq:U_time_evo}
U^{a,\alpha}_r = \exp(-iO^{a,\alpha}_r\sqrt{\tau}).
\end{equation}

The channel described by Eq.~\eqref{eq:c_channel} can be implemented using the following procedure:
$$
\begin{array}{c}
\Qcircuit @C=1em @R=1.5em {
    & \lstick{|0\rangle} & \multigate{1}{\quad U^{a,\alpha}_r\quad} & \qw & \rstick{\text{Discard.}} \qw \\
    & \lstick{\rho}      & \ghost{\quad U^{a,\alpha}_r\quad}        & \qw & \qw
}
\end{array}
$$
This subcircuit acts on a single ancilla qubit initialized in the state $|0\rangle$ and the subsystem supporting the jump operators. In architectures with a limited number of ancilla qubits, reset operations can be utilized to obtain a fresh supply of ancilla qubits.

 This single-ancilla based construction can be used to approximate the Trotter steps up to an additional error comparable to the Trotter error itself, as follows from the following lemma.
\begin{lem}\label{lem:Quantumgate} The channel in Eq.~\eqref{eq:c_channel} satisfies
\be\label{eq:gadget_error}
\left\|\mathcal C^{a,\alpha}_{r,\tau}-\exp(\tau \mathcal{L}_{a,\alpha}^{\beta,r})\right\|_\diamond = O(\tau^2).
\ee
\end{lem}
This result was previously shown without the coherent part in Ref.~\cite{chen2024randomizedmethodsimulatinglindblad}, while a detailed proof of the lemma is provided in \cref{app:ProofQuantumgate}. The proof is based on a standard Taylor expansion argument \cite{lang1993real}. Notably, the error in Eq.~\eqref{eq:gadget_error} can be further reduced to $O(\tau^k)$ using $k$ ancilla qubits~\cite{Ding2024Simulating}.
As a corollary, the compilation of channels $\exp(\tau \mathcal{L}^{\beta,r}_{a,\alpha})$ reduces to compilation of the unitaries in Eq.~\eqref{eq:U_time_evo}. 

Compiling the unitary involves designing a circuit for a small subsystem of a $D$-dimensional lattice, encompassing $\Theta(r^D)$ qubits. According to Corollary~\ref{cor:trunc}, constructing a circuit that exactly reproduces the Gibbs distribution requires polynomial classical resources in one dimension, but grows quasipolynomially with dimension for $D > 1$. In contrast, when the goal is to estimate local observables, Corollary~\ref{col:loc_obs} shows that local channels can be constructed using only a constant number of qubits. The latter results in a highly efficient compilation process from a complexity-theoretic standpoint.

However, from a practical perspective, the runtime of standard compilers may be prohibitive, and we may want even shorter circuits, tailored to a specific architecture and gate set. To achieve this, we can \textit{variationally} compile the necessary local unitaries resulting from dilation. 
Specifically, we approximate the target unitary $U^{a,\alpha}_r$ with a parameterized template circuit $V(\bm{\theta}^{a,\alpha}_r)$, 
where $\bm{\theta}^{a,\alpha}_r$ denotes the set of variational parameters defining the ansatz circuit \cite{cerezo2021variational,Tepaske2023Optimal}. An example of such a circuit for a one-dimensional system with $r = 1$ is shown in Fig.~\ref{fig:Ansatz}(a). This ansatz employs entangling two-qubit $CZ$ gates interleaved with parameterized single-qubit gates of the form
\begin{align}\label{eq:singlequbit}
    u(\theta, \phi, \lambda) =
    \begin{pmatrix}
    \cos\left(\frac{\theta}{2}\right) & -e^{i\lambda}\sin\left(\frac{\theta}{2}\right) \\
    e^{i\phi}\sin\left(\frac{\theta}{2}\right) & e^{i(\phi+\lambda)}\cos\left(\frac{\theta}{2}\right)
    \end{pmatrix},
\end{align}
where $\theta$, $\phi$, and $\lambda$ are free parameters comprising the vector $\bm{\theta}^{a,\alpha}_r$. Further decomposition of these single-qubit gates into native hardware-specific operations can be performed straightforwardly.

To determine the optimal parameters, we use minimization
\begin{equation}\label{eqs:optimization}
\bm{\theta}^{a,\alpha}_r = \arg\min_{\bm{\theta}} C^{a,\alpha}_r(\bm{\theta})
\end{equation}
of the loss function $C^{a,\alpha}_r(\bm{\theta})$ that quantifies the discrepancy between the target and the implemented quantum channels. For systems with translational symmetry, it is sufficient to optimize the resulting operation for each unique local configuration, drastically reducing the overall number of independent optimization tasks. The choice of loss function is important to reduce classical simulation overhead. 

In this work, we consider a modified squared-Frobenius norm for the relevant matrix elements,
\begin{align}\label{eq:lossfunction}
C^{a,\alpha}_r(\bm{\theta}) := \sum_i \sum_{j \leq 2^{k-1}} \left| \left[U^{a,\alpha}_r\right]_{i,j} - \left[V(\bm{\theta})\right]_{i,j} \right|^2,
\end{align}
where $A_{ij}$ corresponds to the matrix element of matrix $A$ and $k$ is the total number of qubits in the gadget. Importantly, it is not necessary to match the entire unitary. Since the ancilla qubit is always initialized in the state $\ket{0}$ and discarded at the end, only the corresponding subspace of the unitary contributes to the final outcome.
We justify this choice of loss function more in \cref{app:loss_function_details}.

In general, representing an arbitrary quantum channel requires an exponential number of parameters in Eq.~\eqref{eqs:optimization} as a function of the radius $r$, which could in principle make the optimization costly. However, one may note that the channels appearing in Eq.~\eqref{eq:gadget_error}, for small finite evolution times $\tau$, admit accurate approximations via a Magnus expansion in terms of jump operators in Eq.~\eqref{eq:trunc_LG}. Since each jump operator can be generated with a combination of linear-depth Hamiltonian simulation circuits by construction, each jump operator and the channel itself exhibits polynomial rather than exponential complexity in terms of radius $r$. We therefore expect that, in practice, optimization landscapes are substantially more benign than in the fully generic case. This expectation is corroborated by the results in the following section, which demonstrate that comparatively low-depth circuits—significantly shallower than those typically required for generic channels—are sufficient to capture the relevant dynamics.

\section{Numerical studies}\label{sec:Numerics}

In this section, we evaluate the accuracy of the proposed Gibbs state preparation method using exact numerical simulations of two types of nonintegrable Hamiltonians: mixed-field Ising spin chain and two-dimensional transverse-field Ising model. Our results indicate that a relatively small truncation radius (\( r \leq 3 \)) and a moderate number of Trotter steps (up to 100) are sufficient to ensure rapid convergence of local observables, such as energy density and correlation functions, to their thermal expectation values as determined by exact diagonalization. These findings demonstrate that the method operates within the capabilities of what is commonly referred to as noisy intermediate-scale quantum (NISQ) devices. In the following section, we will examine how various noise levels affect these simulations.

The Hamiltonian of the mixed‐field Ising spin chain is
\begin{align}\label{eq:MFI}
  H = \sum_{(i,j) \in \mathcal S_1} S^z_i S^z_{j}
    + g \sum_{i=1}^{n} S^x_i
    + h \sum_{i=1}^{n} S^z_i,
\end{align}
where $\mathcal S_1$ is a set of nearest neighbors on a one-dimensional ring with periodic boundary conditions,  $S^x_i := \frac 12 X_i$ and  $S^z_i := \frac 12 Z_i$ are spin-1/2 operators. We choose \(g=(\sqrt{5}+5)/8\approx0.9045\), and \(h=(\sqrt{5}+1)/4\approx0.809\). These parameters ensure no single term dominates the energy spectrum, placing the system deep in the nonintegrable (chaotic) regime, with ballistic operator spreading and rapid entanglement growth~\cite{Kim_2011Ballistic}.  Such strong transport dynamics aims to shorten the thermalization time. In \cref{sec:Further models}, we also show that our truncated Lindbladian method remains effective for examples of integrable models and systems with $U(1)$ symmetry. 

The two-dimensional transverse-field Ising model is given by
\begin{align}\label{eq:2dIsing}
    H = \sum_{(i,j) \in \mathcal S_2} S_i^zS^z_j+g'\sum_{i=1}^n S^x_i,
\end{align}
where $\mathcal S_2$ contains the nearest neighbors on a square lattice with periodic boundary conditions. We choose $g' = 0.2$. The model is known to exhibit a thermal phase transition~\cite{Hesselmann_2016}.

We choose the details of the dynamics as follows. First of all, we study quantum Gibbs state preparation starting from a random product state. This choice is equivalent to setting $\rho(0) = \rho_{\beta = 0}$, i.e. as the maximally mixed, infinite-temperature state. Thus, one may consider it as an energy-reducing, ``cooling" process. The non-unitary part of dynamics is generated by the transformed Pauli $X$, $Y$, and $Z$ operators acting on each qubit (see Eq.~\eqref{jump_operator_choice}) constructed following Eq.~\eqref{eq:trunc_LG} for a fixed truncation radius $r$.  Unless stated otherwise, we use a Gaussian envelope function $q(\nu)$ defined in Eq.~\eqref{eq:Gauss}.

As the target of our numerical study, we analyze how different parameter choices (e.g. Trotter step or truncation radius) impact the performance of the algorithm. As the main precision metric, we use the Gibbs state energy density and its the deviation from the true value,
\begin{align}\label{eq:deltaE}
&E(t) := \frac{1}{n} \tr(\rho(t) H)\nonumber\\
&\Delta E(t) := \frac{1}{n} \Bigl| \tr(\rho(t) H) - \tr(\rho_\beta H) \Bigl|,
\end{align}
where $\rho(t)$ is the result of the evolution
\begin{equation}
\rho(t) = \mathcal{E}_{t,\tau}(\rho(0)),
\end{equation}
with $\mathcal{E}_{t,\tau}$ defined in Eq.~\eqref{eq:average} for a fixed Trotter step $\tau$, and time taking discrete values $t = M\tau$.

We also analyze the algorithm's performance in reproducing  other local observables. To test this, we consider the correlation function
\be\label{eq:corrf_def}
\delta(\bm{a},t) := \tr[\rho(t) S^z_{a_1}S^z_{a_2}]-\tr[\rho(t) S^z_{a_1}]\tr[\rho(t) S^z_{a_2}].
\ee
where $\bm{a} = (a_1,a_2)$. This function characterizes the emergence of spin-spin correlations within the system. Since the one-dimensional mixed-field Ising model does not exhibit an equilibrium ferromagnetic phase, we expect the correlator to decay exponentially with distance, i.e., \( \delta(\bm{a}, t) \sim \exp(-\ell(a_1,a_2) / l_0) \), where \( \ell(a_1,a_2) \) is the spatial separation between $a_1$ and $a_2$ and \( l_0 > 0 \) denotes the correlation length.

In order to detect the phase transition in the two-dimensional model, we analyze the heat capacity
\begin{align}\label{eq:heatcapacity}
    C_{\beta}:= \beta^2 \left(\tr[\rho(t)H^2]-\tr[\rho(t)H]^2\right).
\end{align}
In Sections~\ref{sec:trunc_numer}--\ref{subsec:Envelope}, we study the performance of the algorithm without accounting for compilation errors. 
In Sections~\ref{sec:compile_gadget} and \ref{sec:noise_numer}, we further compile each gadget into two-qubit gates, assuming the architecture illustrated in Fig.~\ref{fig:Ansatz}. In all cases, we compare the results to numerical values obtained via exact diagonalization, i.e.\ the ``ground truth''.

\subsection{Effect of truncation}
\label{sec:trunc_numer}

 \begin{figure*}[t!]
\centering
\includegraphics[width=1.0\textwidth]{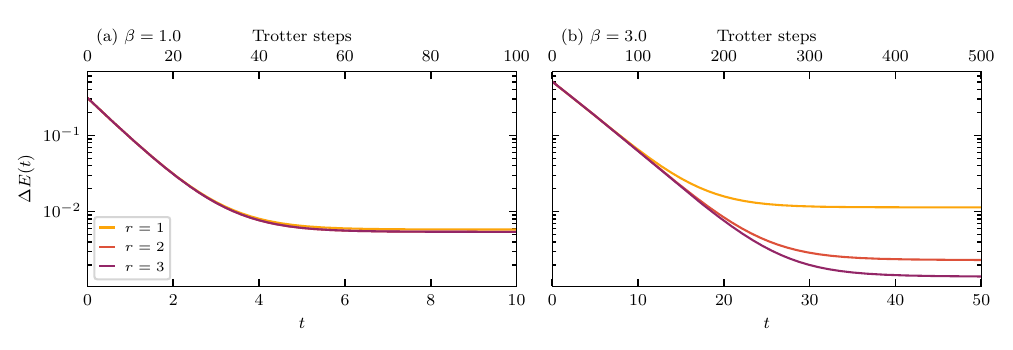}
\caption{
  \textbf{Effect of truncation on time convergence.}
  Absolute error in the internal energy density $\Delta E(t)$ in Eq.~\eqref{eq:deltaE} for the mixed-field Ising Hamiltonian in Eq.~\eqref{eq:MFI} on a chain of $n=12$ spins. Dynamics are generated by direct mixed-state simulations of the Trotterized channel Eq.~\eqref{eq:trotterized} with Trotter step $\tau = 0.1$, for a system initialized in the maximally mixed state. The $x$-axis markers on the bottom denote the physical time $t$, while those at the top indicate the number of Trotter steps $M = t/\tau$. Curves correspond to truncation radii $r=1$ (yellow), $r=2$ (orange), and $r=3$ (red). 
  (a) Inverse temperature $\beta = 1$. 
  (b) Inverse temperature $\beta = 3$.
}
\label{fig:Plot_time_conv}
\end{figure*}

\Cref{fig:Plot_time_conv} shows how our protocol cools the system from an infinite-temperature initial state to thermal states at two values of temperature using a Trotter step $\tau=0.1$.  At the higher temperature ($\beta=1$), the internal energy already achieves best possible value with the relative error $\Delta E(t)/|E(t)| \sim 10^{-2}$ for the minimal truncation radius $r=1$, and increasing $r$ yields no noticeable improvement. By contrast, at the lower temperature ($\beta=3$), the convergence depends strongly on $r$. We attribute this behavior to the enhanced spin–spin correlations in the mixed-field Ising model, which become more significant at lower energies. At the same time, the cooling rate, as measured by the slope of the exponential decay at early times, is largely insensitive to the truncation radius.
 A bigger picture across temperatures is shown in Fig.~\ref{fig:Plot_temp_conv}. These plots illustrate that the energy convergence is nonmonotonic across temperatures but rapidly converges with respect to truncation radius at low temperatures. 
 
The rapid convergence with respect to the truncation radius is further demonstrated by local observables other than energy density. In Fig.~\ref{fig:Longrange}, we present the two-point correlation function at $\beta = 3$. While the minimal truncation radius ($r = 1$) yields noticeable deviations from the expected spatial profile, extending the radius to $r = 3$ suffices to accurately reproduce the correct space correlation profile.

These plots indicate that the radius $r$ can be employed as a systematic control parameter, which is increased until the observable of interest converges to a stable value. The results further show, as expected, that lower temperatures require larger values of $r$ to achieve convergence.

 \begin{figure*}[t!]
\centering
\includegraphics[width=1.0\textwidth]{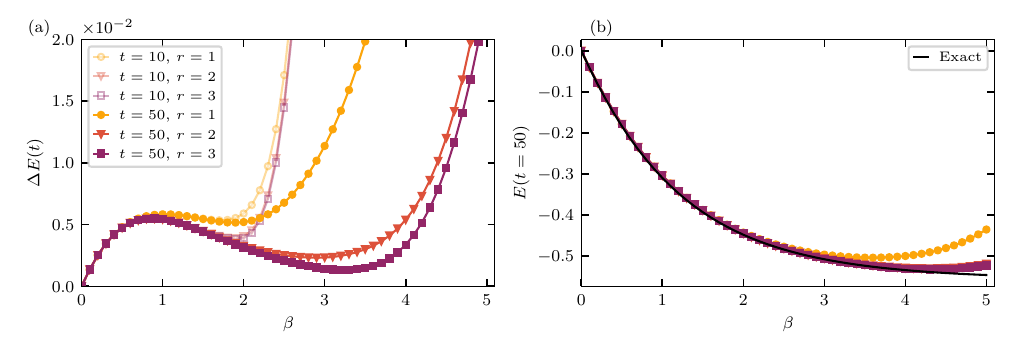}
\caption{
  \textbf{Energy convergence across temperatures}.
  Absolute error in the energy density $\Delta E(t)$ in Eq.~\eqref{eq:deltaE} for the same simulation as in Fig.~\ref{fig:Plot_time_conv}, plotted as a function of inverse temperature $\beta$.
  (a) Error $\Delta E(t)$ at $t = 10$ and $t = 50$ for truncation radii  $r=1$ (yellow), $r=2$ (orange), and $r=3$ (red). The pronounced rise at large $\beta$ reflects mixing times that exceed $t_{\mathrm{mix}} > 50$. 
  (b) Absolute energy density for $t = 50$ and various truncation radii demonstrating convergence toward the true energy density (black curve) that reaches groundstate value $E_0 / n = -0.557$ in the low-temperature (high-$\beta$) regime.
}

\label{fig:Plot_temp_conv}
\end{figure*}

\begin{figure}[t!]
\includegraphics[width=0.5\textwidth]{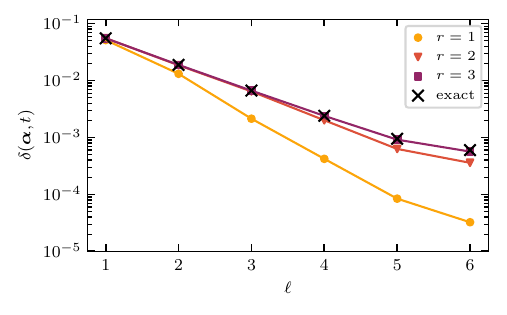}
\caption{
  \textbf{Spatial profile of the two-spin correlators.}
  Two-point correlator \(\delta(\mathbf{a},t)\) in Eq.~\eqref{eq:corrf_def} for spin pair $\mathbf{a} = (n/2,n/2+\ell)$, as a function of spin separation $\ell$ for the mixed-field Ising chain  in Eq.~\eqref{eq:MFI} at inverse temperature \(\beta = 3\) and time \(t = 10\). Curves correspond to truncation radii \(r = 1\) (yellow), \(r = 2\) (orange), and \(r = 3\) (red), while black crosses denote the exact results. The accuracy of the exponential decay of correlations improves with truncation radius.
}
\label{fig:Longrange}
\end{figure}

\subsection{Effect of Trotterization}

A critical parameter in our protocol is the choice of the Trotter time-step, $\tau$. As shown in Fig.~\ref{fig:PlotTrotter}, for parameter regimes relevant to near-term quantum processors ($\beta = 1$, $r = 1$), the Trotterization error constitutes the dominant contribution to the total simulation error. By reducing $\tau$ one can suppress this error by multiple orders of magnitude; intriguingly, however, over the range of $\tau$ values we examined the error does not display clear convergence even up to $\tau\sim10^{-3}$, implying that truncation artifacts remain subdominant. Unfortunately, a smaller $\tau$ necessarily increases circuit depth, which is severely constrained on noisy hardware without full fault tolerance. We speculate that intermediate, early fault-tolerant architectures---where available qubit counts are the primary limitation, rather than circuit depth---will provide an ideal testbed for these Gibbs-state preparation methods.

\begin{figure}[t!]
\includegraphics[width=0.5\textwidth]{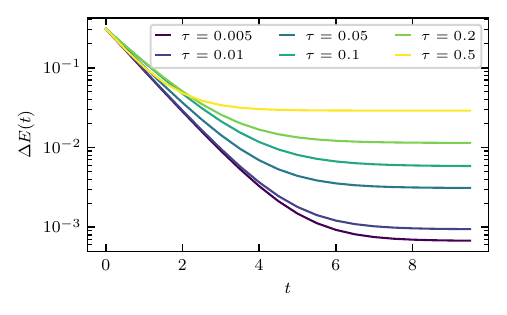}
\caption{
  \textbf{Trotter step dependence.}
  Error in the thermal‐state preparation protocol for the mixed‐field Ising Hamiltonian [Eq.~\eqref{eq:MFI}] at inverse temperature $\beta = 1$, plotted as a function of the Trotter time step $\tau$ for truncation radius $r = 1$. The rapid decrease in error with smaller $\tau$ confirms that Trotterization is the primary source of protocol inaccuracy.
}
\label{fig:PlotTrotter}
\end{figure}

\subsection{Effect of the envelope function choice}
\label{subsec:Envelope}

The envelope function \(q(\nu)\) in Eq.~\eqref{eq:envelope} provides an additional degree of freedom for tuning the precision of the Gibbs‐state preparation protocol. Aside from the constraint \(q(-\nu)=q^*(\nu)\), no further analytic restrictions apply, yet the choice of \(q(\nu)\) markedly affects performance. Here, we compare the default Gaussian envelope of Eq.~\eqref{eq:envelope} against two alternatives:
\begin{subequations}
\begin{align}
    q(\nu) &= 1, \label{eq:constant}\\
    q(\nu) &= \exp\biggl(-\frac{\sqrt{1+(\beta\nu)^2}}{4}\biggr), \label{eq:Metropolis_hasingsfilter}
\end{align}
\end{subequations}
where Eq.~\eqref{eq:Metropolis_hasingsfilter} represents a smoothed Metropolis–Hastings filter introduced in Ref.~\cite{ding2024efficientquantumgibbssamplers}, and Eq.~\eqref{eq:constant} is a constant (flat) envelope.  In each case, we additionally renormalize the jump operators to preserve their Frobenius norm relative to the Gaussian filter (see \cref{subsec:Normalization} for details). 
Although a flat filter formally leads to a divergent filter function $f(t)$, the jump operators and coherent terms in Eq.~\eqref{eq:trunc_LG} remain well defined operator on the support of $H_{a,r}$. The numerical tests on local observables show that it surprisingly yields the highest precision among the three envelopes. 

  In Fig.~\ref{fig:Filterfunction} we compare the cooling dynamics and residual errors: the flat envelope \(q(\nu)=1\) delivers the fastest initial cooling and the smallest asymptotic error, whereas the smoothed Metropolis–Hastings envelope exhibits the slowest convergence and largest errors, with this discrepancy growing at larger inverse temperatures, see Fig.~\ref{fig:Filterfunction}(b).  These findings indicate that the envelope choice is a powerful lever for optimizing Gibbs‐state preparation and merits a systematic study in future work.

\begin{figure}[t!]
\includegraphics[width=1.0\columnwidth]{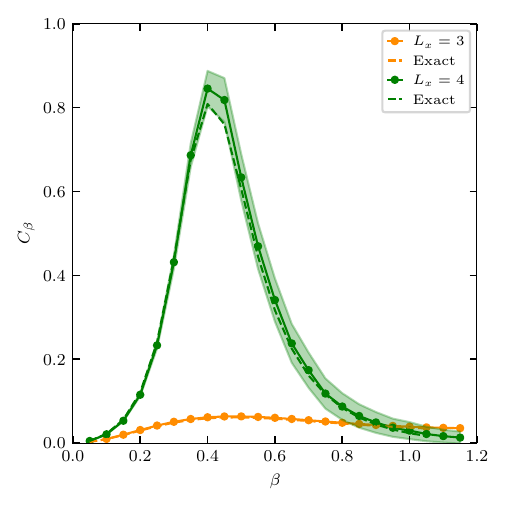}
\label{fig:heatcapacity}
\caption{\textcolor{black}{\textbf{Signatures of thermal phase transition in 2D transverse field Ising model.} The heat capacity $C_\beta$~[cf. Eq.~\eqref{eq:heatcapacity}] as a function of inverse temperature $\beta$ for the 3$\times$3 (orange) and 4$\times$4 (green) square lattices with periodic boundary conditions. Comparison between exact data~(dashed lines) and simulation~(dots, solid lines). The results for the 4$\times$4 lattice were obtained using sampling over 16 384 state-vector simulations, the shaded region indicates one standard deviation. The results show that our protocol is  able to reproduce the peak in the heat capacity for finite systems.}}
\end{figure}

\begin{figure*}[t!]
\includegraphics[width=1.0\textwidth]{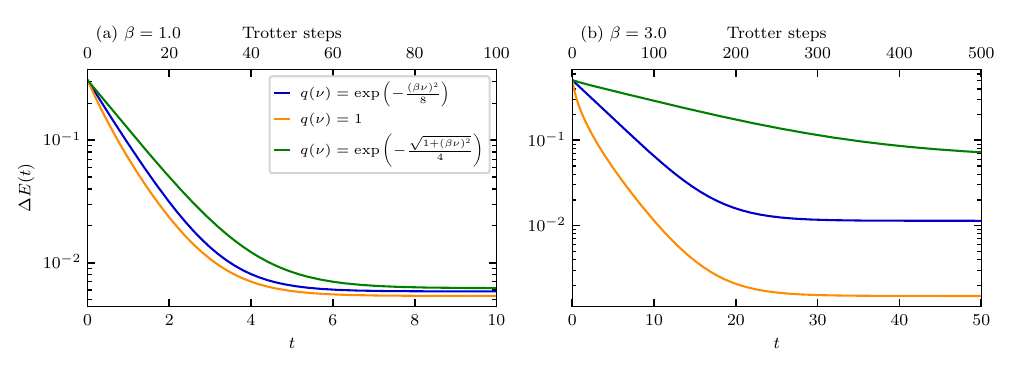}
\caption{\textbf{Impact of the envelope function.} Dynamics of our thermal‐state preparation protocol with Trotter time‐step \(\tau = 0.1\). For this plot only, jump operators and coherent terms in all cases are renormalized to Frobenius norm of the Gaussian filter for a consistent comparison, see \cref{subsec:Normalization}. (a) At inverse temperature \(\beta = 1\), the flat envelope \(q(\nu)=1\) achieves the fastest convergence and lowest steady‐state error. (b) At \(\beta = 3\), the performance gap widens, highlighting that an optimal choice of \(q(\nu)\) can dramatically accelerate the performance of the Gibbs sampler.}
\label{fig:Filterfunction}
\end{figure*}

\begin{figure*}[t!]
\includegraphics[width=1.0\textwidth]{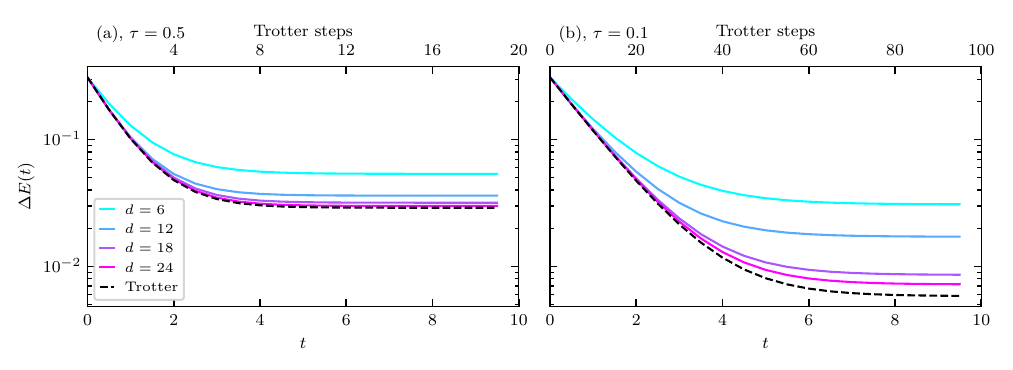}
\caption{\textbf{Compilation error.}
Error in circuit implementation using the randomized compilation strategy described in Fig.~\ref{fig:Trotterization}(b), where each elementary operation is approximated using the template from Fig.~\ref{fig:Ansatz}(b). Results are shown for circuit depths $d = 6, 12, 18$, and $24$, corresponding to $m = 2, 4, 6$, and $8$ modules, respectively. The results are compared to Trotterized evolution with no compilation error (dashed curve). Simulations are performed at inverse temperature $\beta = 1.0$ and final time $t = 10$. Panels show results for Trotter step sizes: (a)~$\tau = 0.5$ and (b)~$\tau = 0.1$. The results show that for $\tau=0.5$, a depth of $d=12$ suffices to render the compilation error negligible compared to the Trotterization error. For the smaller step size $\tau=0.1$, a depth of $d=18$ yields an energy density error of approximately $\Delta E(t) \sim 10^{-2}$ on a 12-qubit system.}
\label{fig:CompilationResults}
\end{figure*}

\subsection{Two-dimensional system}
\label{sec:2d_sim}

Finally, we present results for the two-dimensional transverse-field Ising model defined in Eq.~\eqref{eq:2dIsing}, demonstrating that our protocol is capable of reproducing thermal phase transitions in finite systems. Specifically, we compute the heat capacity $C_\beta$ [cf.\ Eq.~\eqref{eq:heatcapacity}] for simulations on $3\times 3$ and $4\times 4$ lattices. For the $4\times 4$ system, we numerically implement the protocol by averaging over $16\,384$ stochastic trajectories. We fix the truncation radius to $r=1$, choose a Trotter step size $\tau=0.1$, and evolve the system up to a final time $t=20$.

In these simulations, we empirically find that replacing the commonly used filter function in Eq.~\eqref{eq:Gauss} with the temperature-independent Gaussian filter
\begin{align}
	q(\nu) = \exp\!\left(-\frac{\nu^2}{2}\right)
\end{align}
leads to significantly improved convergence toward the exact results at low temperatures. The underlying intuition is that the filter function must retain frequency components comparable to the Hamiltonian gap, which may lie outside the passband of standard temperature-dependent filters in the low-temperature regime.

The resulting heat-capacity data are shown in Fig.~\ref{fig:heatcapacity}. The presence of a broad peak in $C_\beta$ signals a crossover between the ferromagnetic and paramagnetic phases. These results demonstrate that the presence of a phase-transition boundary does not pose a fundamental limitation to the present algorithm.

\section{Hardware-aware simulations}
\label{sec:hardware_sim}

In \cref{sec:Numerics} we characterized the performance of the exact Trotterized channels, neglecting any overhead from compilation and hardware noise. In this section, we incorporate realistic circuit‐compilation constraints and noise models to assess their impact on our Gibbs‐state preparation protocol.

\subsection{Effect of circuit compilation}
\label{sec:compile_gadget}

Our objective is to show that thermalizing Lindbladian evolution can be approximated with a handful of circuits with a manageable two-qubit gate count compatible with the current noisy quantum devices.

The hardware compilation procedure proceeds as follows. We begin by approximating the Lindblad evolution via Eq.~\eqref{eq:average}. Each of these circuits is constructed by substituting every elementary exponential in Eq.~\eqref{eq:channelrandomvariable} with its approximate representation given in Eq.~\eqref{eq:c_channel}. The unitaries $U^{a,\alpha}_r$ are compiled using circuit templates illustrated in Fig.~\ref{fig:Ansatz}(b), where the parameters are optimized via the Adam algorithm~\cite{adam_kingma_2014} to minimize the loss function defined in Eq.~\eqref{eqs:optimization}. We investigate circuit architectures comprising $m = 2, 4, 6,$ and $ 8$ modules, which correspond to two-qubit gate depths of $d = 6, 12, 18,$ and $24$, respectively. For each configuration, the optimization is performed over 50 random parameter initializations, with 8000 optimization steps per instance. The parameters yielding the lowest final loss are selected for compilation. Further implementation details are provided in Section~\ref{optimiz_scetion}.

To isolate the effects of gate compilation, we compare the time dynamics of the compiled circuits to those of the ideal Trotterized evolution, which employs the exact unitary $U^{a,\alpha}_r$ in Eq.~\eqref{eq:U_time_evo}. We assess the performance of the thermal state preparation protocol at inverse temperature $\beta=1$ for two Trotter step sizes, $\tau=0.5$ and $\tau=0.1$. The results, shown in Fig.~\ref{fig:CompilationResults}, indicate that for $\tau=0.5$, a depth of $d=12$ suffices to render the compilation error negligible compared to the Trotterization error. For the smaller step size $\tau=0.1$, a depth of $d=18$ yields an energy density error of approximately $\Delta E(t) \sim 10^{-2}$ on a 12-qubit system. These depths are well within the capabilities of current experimental platforms~\cite{kim2023evidence} and are significantly lower than the estimated lower bound of $d=64$ for generic four-qubit gate synthesis~\cite{Shende_2006}.

\subsection{Effect of noise}
\label{sec:noise_numer}

\begin{figure*}[t!]
\includegraphics[width=1.0\textwidth]{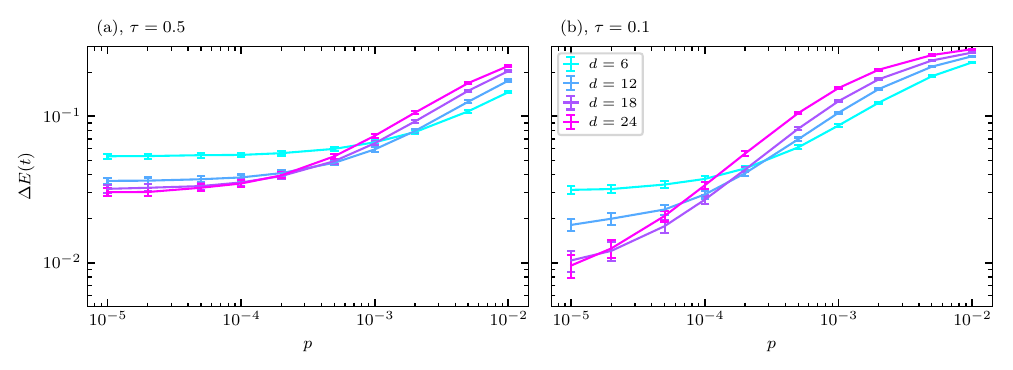}
\caption{\textbf{Effect of noise.} Accuracy of thermal state preparation as a function of the depolarizing noise rate $p$, for inverse temperature $\beta = 1$ and evolution time $t = 10$. Error bars reflect the statistical error arising from $10^3$ random circuits and 1024 shots per circuit. Results are shown for Trotter step sizes: (a)~$\tau = 0.5$ and (b)~$\tau = 0.1$. For target noise rates on the order of $p \sim 10^{-4}$, energy density errors on the order of $\sim 10^{-2}$ are achievable.}
\label{fig:Noisy}
\end{figure*}

Finally, we analyze the influence of noise on the accuracy of the simulations. To emulate realistic conditions, we assume that each $k$-qubit unitary gate is followed by a $k$-qubit depolarizing noise channel,
\begin{align}
    \mathcal{N}_k(\rho) := (1 - p_k)\rho + \frac{p_k}{4^k - 1} \sum_{i=1}^{4^k - 1} P^{(k)}_i \rho P^{(k)}_i,
\end{align}
where $P^{(k)}_i$ denote the $4^k - 1$ non-identity Pauli operators acting on $k$ qubits, and $p_k$ is the uniform $k$-qubit noise rate. In our simulations, we consider only one- and two-qubit gates, and set $p_1 = 0.1\,p$ and $p_2 = p$, where $p$ is a global noise parameter. This reflects the empirical observation that two-qubit gates are typically the dominant source of error in current quantum hardware, whereas single-qubit gates have an order of magnitude less error.

Each data point in our analysis includes statistical uncertainty arising from limited number of $10^3$ independently generated circuits corresponding to random trajectories $\bm{\alpha} \sim \Sigma^{\otimes n}$ and a finite number of measurement shots, with 1024 shots per circuit. This amounts to approximately $10^6$ shots per data point run on a hypothetical quantum processor.

The simulation results are presented in Fig.~\ref{fig:Noisy}. We observe a distinct crossover in the accuracy as a function of the noise rate $p$, with the crossover point $p_0$ dependent on the Trotter step size $\tau$. For noise levels $p > p_0$, the error is dominated by noise, and increasing the circuit depth $d$ leads to worse performance. Conversely, in the regime $p < p_0$, the primary contribution to the error originates from imperfect compilation, which improves with increasing depth. This crossover behavior is transient: as the circuit depth increases, $p_0$ shifts toward lower values, consistent with the expectation that compilation error vanishes in the large-depth limit, leaving noise as the sole limiting factor. Our results further suggest that for target noise rates on the order of $p \sim 10^{-4}$, energy density errors on the order of $\sim 10^{-2}$ are achievable.

\section{\label{sec:Discussion}Discussion}

This work introduces a framework for quantum Gibbs state preparation using local quantum circuits, extending the qQCMC framework developed in Refs.~\cite{chen2023efficient,ding2024efficientquantumgibbssamplers}. By integrating operator truncation, Trotterization, and variational circuit compilation, we demonstrate that the protocol can be adapted for implementation on near-term quantum hardware. Our results show that the steady state of the truncated Lindbladian closely approximates the true Gibbs state, both in trace distance and in expectation values of local observables, with errors that can be systematically reduced by increasing the truncation radius. In addition, the randomized Trotterization strategy ensures low circuit depth while preserving theoretical convergence guarantees.

A key strength of our approach lies in its flexibility. The truncation scheme and circuit ansatz can be adapted to arbitrary lattice geometries and hardware topologies, making the method broadly applicable beyond the one-dimensional systems studied numerically in this work. Moreover, the use of variationally optimized circuit gadgets enables efficient compilation even under restricted gate sets and limited qubit connectivity. This makes the protocol more practical for experimental implementation compared to earlier proposals~\cite{chen2023efficient,ding2024efficientquantumgibbssamplers}, especially in settings where only local operators and a small number of ancilla qubits are available.

Several promising directions for future investigation emerge from this study. One natural question concerns the scaling of the required truncation radius at low temperatures, particularly near phase transitions. While our current focus is on the high-temperature regime, the protocol may be extensible to weakly interacting fermionic systems at arbitrary temperatures, in light of recent theoretical advances~\cite{vsmid2025polynomial,tong2024fast}. Another important consideration is noise: although our numerical results suggest robustness to moderate levels of depolarizing noise, further improvements may be possible by incorporating noise models directly into the gadget optimization process.

Additionally, our investigation of different filter functions reveals that algorithmic performance can be substantially improved through appropriate parameter tuning. A systematic understanding of how performance depends on the choice of filter remains an open question. Recent work on quasi-particle cooling~\cite{Quasiparticle2025Lloyd} suggests that algorithmic efficiency may depend sensitively on the structure of low-energy excitations. Our results pave the way for studying this dependence in greater depth and may ultimately offer new insights into the nature of thermalization processes.

Finally, an important open question is whether this approach can provide a viable route to quantum advantage on near-term or early fault-tolerant setting. Specifically, future work must establish whether low-noise, finite-size quantum processors can achieve sample complexity that surpasses both established classical algorithms for simulation of quantum thermal states \cite{alhambra2023quantum} and leading tensor-network methods of the relevant Lindbladian dynamics, e.g. proposed in \cite{zhan2025rapid}.

\textbf{Acknowledgments}. We thank Alvaro Alhambra, Anirban Chowdhury, Alexander Miessen, Bibek Pokharel, Kunal Sharma, and Derek Wang for useful discussions. RS is grateful to the Alexander von Humboldt Foundation for support, via the German Research Chair program. DH acknowledges support  from the UK EPSRC under
grant EP/X030881/1, and from a Leverhulme Trust International Professorship grant (Award Number:
LIP-2020-014, for a Leverhulme-Peierls Fellowship at Oxford).

\appendix
\section{List of Notations}

\label{sec:notations}

In this Appendix, we collect the main symbols, operators, and definitions used throughout the paper. We use lowercase Latin and Greek letters for numbers, numerical functions, labels, and indices (e.g., $n$ or $\beta$), uppercase Greek letters for sets (e.g., $\Lambda$), uppercase Latin letters for operators (e.g., $H$), except for density matrices, which are traditionally denoted by $\rho$. Calligraphic uppercase letters represent superoperators (e.g., $\mathcal{L}$), while sans-serif letters denote linear spaces and algebras (e.g., $\mathsf{H}$).
\\

\subsection{General notation}

\begin{itemize}[label={}, leftmargin=5pt]
  \item $n \in \mathbb{N}$ -- Number of qubits (system size). This number does \textit{not} include ancilla qubits.
  
  \item \textbf{$\Lambda = \{a_1, \dots, a_n\}$} -- A finite $D$-dimensional square lattice with $n$ sites equipped with a graph structure.

\item \textbf{$\ell(a,b)$} -- The graph distance defined by
$$
\ell(a,b):=\min_{\Pi=\{a\to c_1,\dots,\,c_k\to b\}}|\Pi|,
$$
where the minimum is taken over all paths $\Pi$ between $a\in\Lambda$ and $b\in\Lambda$.

\item $B_a(r) = \{b \mid \ell(a,b) \leq r\}$ — The subset of lattice sites located within a ball of radius $r\geq 0$ centered at site $a$.
  
  \item \textbf{$\mathsf{H}_{\{a\}} \cong \mathbb{C}^2$} -- Local Hilbert space associated with qubit at site $a \in \Lambda$.
  
  \item \textbf{$\mathsf H_X = \bigotimes_{a \in X} \mathsf{H}_{\{a\}}  \cong (\mathbb{C}^2)^{\otimes \abs{X}}$}  -- Hilbert space of the qubits subset $X\subseteq \Lambda$.
  
  \item \textbf{$\mathsf M_X = \operatorname{End}(\mathsf H_X)$} -- Algebra of linear operators acting on $\mathsf H_X$ for $X\subseteq \Lambda$.

    \item $\mathsf{S}_X = \{O \in \mathsf{M}_X \mid O=O^\dag\}$ -- The set of self-adjoint (Hermitian) operators.

    \item $\mathsf{D}_X = \{\rho \in \mathsf{S}_X \mid \rho \geq 0,\ \tr \rho = 1\}$ — The set of density matrices representing mixed quantum states of the qubits.
    
  \item \textbf{$H \in \mathsf S_\Lambda$} -- Local Hamiltonian on the lattice $\Lambda$; defined in Eq.~\eqref{eq:ham}.
  
  \item $\lambda_i \in \mathbb{R}$ -- Eigenvalues (energies) of the Hamiltonian $H$, indexed by $i \in \{1, \dots, 2^n\}$.
  
  \item $P_i \in \mathsf S_\Lambda$ -- Projector onto the energy eigenspace of eigenvalue $\lambda_i$, satisfying $P_i H P_i = \lambda_i P_i$.
  
  \item \textbf{$\beta \in \mathbb{R}_{\geq 0}$} -- Inverse temperature.
  
  \item \textbf{$\rho_\beta \in \mathsf D_\Lambda$} -- Gibbs (thermal) state associated with $H$ and inverse temperature $\beta$, see Eq.~\eqref{Eq:Gibbs state}.

  \item  $\mathsf F_\Lambda^{(k,l)}\subset \mathsf S_\Lambda$ -- Family of $(k,l)$-local Hamiltonians on lattice $\Lambda$ being a sum of terms that has no support on more than $k$ sites and  each site appears on at most $l$ terms; see Definition~\ref{def:kl_local}.

\end{itemize}

\subsection{Norms and distances}

\begin{itemize}[label={}, leftmargin=5pt]
  
  \item $\|\cdot\|_1$ -- Schatten 1-norm (trace norm) on operators. For $O \in \mathsf M_\Lambda$, it is defined by
  \begin{equation}
  \|O\|_1 := \operatorname{Tr} \sqrt{O^\dagger O}.
  \end{equation}
\item $\|\cdot\|_\infty$ -- $\infty$-norm (operator or spectral norm) on operators. For $O \in \mathsf M_\Lambda$, it is defined by
  \begin{equation}\label{eq:inf_norm}
  \|O\|_\infty := \sup_{\norm{v}=1} \norm{O v}.
  \end{equation}
  Here $v\in \mathsf H_\Lambda$ and $\norm{v}=\sqrt{\braket{v|v}}$. 
\item $\|\cdot\|_\diamond$ — Diamond norm on superoperators. For $\mathcal E: \mathsf M_\Lambda \to \mathsf M_\Lambda$ it is defined as
\begin{equation*}
\left\| \mathcal E \right\|_\diamond := \max_{X : \left\| X \right\|_1 \leq 1} \left\| \left( \mathcal E \otimes \mathcal I \right) X \right\|_1, \notag
\end{equation*}
where $\mathcal I: \mathsf M_\Lambda \to \mathsf M_\Lambda$ is the identity operator, and the maximum is taken over $X\in \mathsf M_\Lambda^{\otimes 2}$.
\item  $\vertiii{X}$---oscillator norm
\begin{align*}
\vertiii{X}:=\sum_{a\in\Lambda}\|\delta_a(X)\|_\infty\,,
\end{align*}
where $\delta_a(X):=X-\frac{1}{2}I_a\otimes \operatorname{tr}_a(X)$
\item $\|\cL\|_{p\rightarrow p}$: Induced $p-p$ norm on a superoperator $\cL$:
\begin{equation*}
    \|\cL\|_{p\rightarrow p}=\sup_{O} \frac{\norm{\cL[O]}_p}{\norm{O}_p}
\end{equation*}
\end{itemize}
\subsection{Dissipative dynamics}

\begin{itemize}[label={}, leftmargin=5pt]
  \item $t \in \mathbb{R}_{\geq 0}$ -- Continuous time.
  
  \item \textbf{$\mathcal{L}^\beta: \mathsf M_\Lambda \to \mathsf M_\Lambda$} -- Lindbladian superoperator governing the dissipative dynamics that prepares the Gibbs state $\rho_\beta$; see Eq.~\eqref{eq:LindbladLin}.
  
  \item \textbf{$A^{a,\alpha} \in \mathsf M_\Lambda$} -- Single-qubit operator acting on site $a$; chosen to be Pauli operators, see Eq.~\eqref{eq:jump}.
  
  \item $\alpha \in \mathbb{N}$ or $\mathbb R$ -- Discrete index labeling jump operators defined within the neighborhood of qubit $a\in \Lambda$.
  
  \item $f(t),\, q(\nu): \mathbb{R} \to \mathbb{C}$ -- Filter and envelope functions, respectively; see Eq.~\eqref{eq:envelope}.
  
  \item $H_{a,r} \in \mathsf S_\Lambda$ -- Truncated Hamiltonian for the neighborhood of qubit $a$ with radius $r$; see Eq.~\eqref{eq: trunc_ham}.
  
  \item $\mathcal{L}^{\beta,r}: \mathsf M_\Lambda \to \mathsf M_\Lambda$ -- Truncated Lindbladian by the radius $r$; see Eq.~\eqref{eq:trunc_lindbl}.

  \item $\rho_{\beta,r}$ : steady state of truncated Lindbladian $\mathcal{L}^{\beta,r}$, introduced in Theorem~\ref{thrm:trunc} 
  \item $t_{\mathrm{mix}} \in \mathbb{R}_+$ -- Mixing time, i.e. minimum time after which the distance between any two states contracts by at least a certain fixed fraction; see Definition~\ref{dfn:mixing_time}.
\end{itemize}

\vspace{0.5em}
\subsection{Trotterization and compilation:}

\begin{itemize}[label={}, leftmargin=5pt]
  \item $\tau \in \mathbb{R}_+$ -- Trotter step size (time dicretization). Introduced in Section~\ref{subsec:Randomizedcompilation}.
  
  \item $M \in \mathbb{N}$ -- Number of Trotter steps, such that $t = M\tau$. Introduced in Section~\ref{subsec:Randomizedcompilation}.
  
  \item $\mathcal{C}^{a,\alpha}_{r,\tau}: \mathsf M_\Lambda \to \mathsf M_\Lambda$ -- Quantum channel that approximates \(\exp(\tau \mathcal{L}_{a,\alpha}^{\beta,r})\) using one ancilla qubit and a unitary operation. Introduced in Eq.~\eqref{eq:c_channel}.
  
  \item $\mathcal{E}_{t,\tau}: \mathsf M_\Lambda \to \mathsf M_\Lambda$ -- Randomized Trotterization channel for evolution time $t$ compiled as a sequence of $M = t/\tau$ Trotterization gadgets with Trotter step $\tau$. Introduced in Eq.~\eqref{eq:channelrandomvariable}.
  
  \item $m\in \mathbb N$ -- Number of repeating depth-3 modules in the gadget compilation circuit. Introduced in Fig.~\ref{fig:Ansatz}.
\end{itemize}

\subsection{List of functions}

\begin{itemize}[label={}, leftmargin=5pt]

  \item $f(t),\, q(\nu): \mathbb{R} \to \mathbb{C}$ -- Filter and envelope functions, respectively; introduced in Eq.~\eqref{eq:envelope}.
\item $g_1(t)$, $g_2(t): \mathbb{R} \to \mathbb{C}$ -- Kernel functions in time domain, introduced in Eq.~\eqref{eq:coherenttime}.
\item $g(\nu_1,\nu_2): \mathbb{R} \to \mathbb{C}$ -- Combined kernel function for the coherent term in frequency domain, introduced in Eq.~\eqref{eq:G_kernel}
\item $\Delta(r_0)$ function to bound the quasilocality of the Lindbladian, introduced in Eq.~\eqref{eq:defDelta}.
\item $\eta(\beta)$ function to bound the distance of the Lindbladian $\cL^\beta$ to the Lindbladian $\cL^{\beta=0}$, introduced in Eq.~\eqref{eq:defeta}. The analogous function $\eta^\prime(\beta)$ for the truncated Lindbladian $\cL^{\beta,r}$ is introduced in Eq.~\eqref{eq:kappaboundprime}. 
\item  $k(r_0,\beta)$, function to bound the decay of the oscillator norm for the Lindbladian $\cL^\beta$, introduced in Eq.~\eqref{eq:kappabound}. The modified version to characterize the oscillator norm decay for the truncated Lindbladian $\cL^{\beta,r}$, $k'(r_0,\beta)$, is introduced in Eq.~\eqref{eq:kappaboundprime}.
\item $u(r_0)$, function appearing in definition of $k(r_0,\beta)$, defined in Eq.~\eqref{eq:deff0}.The analogous function $u'(r_0)$ for the truncated Lindbladian $\cL^{\beta,r}$ is introduced in Eq.~\eqref{eq:kappaboundprime}. 
\item $\gamma(x)=\frac{\sqrt{x}}{1+\sqrt{x}}$, appearing in bounds for fixed point and local observable errors, first appearance above Eq.~\eqref{eq:abovegamma}.
\item $\xi_1(r,\beta J)$ and $\xi_2(r,\beta J)$ functions appearing in bounds of $\Delta(r_0)$, first defined in Eq.~\eqref{eq:defxi1} and Eq.~\eqref{eq:defxi2}. 
\item $m(|X|)\colon \mathbb{R}\to\mathbb{R}$: quantifies how the truncation-error bound for an operator’s expectation value on $X\subset\Lambda$ grows with the region size $|X|$ (see Theorem~\ref{thrm:trunc_local}).
\item $b(r)\colon \mathbb{R}\to\mathbb{R}$: characterizes the decay of that error bound as the truncation radius $r$ increases (see Theorem~\ref{thrm:trunc_local}).
\end{itemize}

\begin{widetext}

\section{Lieb-Robinson Bounds}\label{sec:Many-body quantum systems}

The error bounds for truncated Lindbladians arise from the locality structure of the underlying lattice, which enables the use of Lieb-Robinson bounds. This Appendix provides a very brief overview of the Lieb-Robinson bounds; for a more detailed introduction, see Ref.~\cite{rouze2024efficientthermalizationuniversalquantum}.

\begin{dfn}\label{def:kl_local} We define $\mathsf{F}_\Lambda^{(k,l)}\subset \mathsf S_\Lambda$ to be the family of $(k,l)$-local Hamiltonians on lattice $\Lambda$ such that any $H\in \mathsf{F}_\Lambda^{(k,l)}$ has the form
$$
H = \sum_{X \subseteq \Lambda} h_{X},
$$
where each term $h_X\in \mathsf S_X$ acts on $X \subseteq \Lambda$ with $|X| \leq k$, and such that each site $a \in \Lambda$ appears in the support of at most $l$ terms in the sum. That is, for all $a \in \Lambda$, the number of subsets $X$ with $a \in X$ and $h_X \ne 0$ is at most $l$.
\end{dfn}

For any Hamiltonian $H \in \mathsf{F}_\Lambda^{(k,l)}$, we define
$$
h := \max_{X} \|h_X\|_\infty
$$
as the uniform upper bound on the norm of the local terms in the Hamiltonian, where $\|\cdot\|_\infty$ is spectral norm, see Eq.~\eqref{eq:inf_norm}.

Assume that the parameters $h$, $k$, and $l$ are independent of the total number of sites $n$.
Under these assumptions, the Lieb-Robinson velocity is defined by
\begin{align}
J := \max_{a \in \Lambda} \sum_{X \ni a} |X| \, \|h_X\|_\infty = O(hkl).
\end{align}
Given a radius $r$, the growth of an operator $A^a \in \mathsf{M}_a$ supported on site $a \in \Lambda$ is bounded by~\cite{lieb1972finite,hastings2004lieb,hastings2006spectral,nachtergaele2006lieb,Haah2023Quantum}
\begin{equation}\label{eq:LiebRobinson}
    \| e^{-iHt} A^a e^{iHt} - e^{-iH_{a,r}t} A^a e^{iH_{a,r}t} \|_\infty
    \le \|A^a\|_\infty \, \frac{(2J |t|)^r}{r!},
\end{equation}
 where $H_{a,r}$ is a reduction of the Hamiltonian to the ball $B_a(r)$ of radius $r$ centered at site $a$, as defined in Eq.~\eqref{eq: trunc_ham}.

\section{Dissipative Dynamics Structure}
\label{sec:time}

This Appendix provides explicit time-domain expressions for the jump operators and the coherent term of the Lindbladian in Eq.~\eqref{eq:lindbladian_main}, which are necessary to provide the proofs in \cref{sec:tracedistanceerror,section:Gap stability}.
\subsection{The Lindbladian in time domain}
\label{subsec:TimreprLindbladians}
Let us focus our attention on the Gaussian envelope function
\begin{equation}
q(\nu) = \exp\left(-\frac{(\beta \nu)^2}{8}\right).
\end{equation}
After taking the Fourier transform, the corresponding filter function from Eq.~\eqref{eq:envelope} becomes
\begin{align}\label{eqs:f_expression}
f(t) &= \frac{1}{2\pi} \int_{-\infty}^\infty \!\rd \nu\,
\exp\left(-\frac{(\beta \nu)^2}{8}\right)
\exp\left(-\frac{\beta \nu}{4}\right)
e^{i t \nu} \nonumber \\
&= \sqrt{\frac{2}{\pi \beta^2}} \,
\exp\left( \frac{(\beta - 4 i t)^2}{8 \beta^2} \right).
\end{align}
The jump operators then have the form
\begin{align}
    L_{a,\alpha}&=\sum_{\nu\in \Omega(H)}q(\nu)e^{-\frac{\beta \nu}{4}}A^{a,\alpha}_\nu\nonumber\\&=\sum_{\nu\in \Omega(H)}\exp\left(-\frac{(\beta \nu + 1)^2 - 1}{8}\right)A^{a,\alpha}_\nu,
\end{align}
where $A^{a,\alpha}_{\nu}$ are spectral components of the operator $A^{a,\alpha}$ as defined in Eq.~\eqref{eq:proj},
and the coherent term in Eq.~\eqref{eq:coherentLin} is
\begin{align}\label{eq:G_kernel}
G_{a,\alpha}^\beta &=-\frac{i}{2}\sum_{\nu\in\Omega(H)} \tanh\left(-\frac{\beta \nu}{4}\right) \left(L^{\beta\dagger}_{a,\alpha} L^\beta_{a,\alpha} \right)_\nu\nonumber\\&= \sum_{\nu_1,\nu_2\in \Omega(H)}
g(\nu_1,\nu_2)
A^{a,\alpha \dagger}_{\nu_2}
A^{a,\alpha}_{\nu_1}.
\end{align}
where the prefactor function is given by
\begin{align}\label{eqs:g_express}
g(\nu_1,\nu_2)
&:= -\frac{i}{2}
\tanh\left(-\frac{\beta (\nu_1 - \nu_2)}{4}\right)
\exp\left(-\frac{(\beta \nu_1 + 1)^2 - 1}{8}\right)
\exp\left(-\frac{(\beta \nu_2 + 1)^2 - 1}{8}\right).
\end{align}
Introducing the shorthand $\nu_\pm = \nu_1 \pm \nu_2$, this simplifies to
\begin{align}\label{eq:kernel_functions}
g(\nu_1,\nu_2)
&= -\frac{i}{2}
\tanh\left(-\frac{\beta \nu_-}{4}\right) \exp\left(-\frac{\beta^2 \nu_+^2 + 4 \beta \nu_+}{16}\right) \exp\left(-\frac{\beta^2 \nu_-^2}{16}\right).
\end{align}
The splitting into contributions in $\nu_+$ and $\nu_-$ allows for conversion into a time representation. Using the results in Appendix A of~\cite{chen2023efficient}, $G_{a,\alpha}^\beta$ can be written as
\begin{align}
G_{a,\alpha}^\beta  = \int_{-\infty}^\infty g_1(t) \, e^{-i H t}  \left( \int_{-\infty}^\infty
g_2(t') \, e^{i H t'} A^{a,\alpha \dagger}
e^{-2i H t'} A^{a,\alpha} e^{i H t'} \rd t' \right) e^{i H t} \rd t,
\end{align}

with the kernel functions $g_1(t)$ and $g_2(t)$ defined as follows Ref.~\cite[Appendix A]{chen2023efficient}.
The first function is
\begin{align}
g_1(t) & =
\frac{1}{\sqrt{2\pi}}
\int_{-\infty}^{\infty} \rd \nu_-\,
e^{i \nu_- t}
\frac{1}{2\pi} \frac{1}{2i}
\tanh\left(-\frac{\beta \nu_-}{4}\right)
e^{- \frac{\beta^2 \nu_-^2}{16} } \nonumber \\
&= 
\left[
\frac{-1}{\pi \beta \cosh\left( \frac{2\pi t}{\beta} \right)}
\right] *_t
\left[
\frac{\sqrt{2}}{\beta}
e^{\frac{1}{4} - \frac{4t^2}{\beta^2}}
\sin\left( \frac{2t}{\beta} \right)
\right],
\end{align}
where $f(t) *_t g(t) = \int_{-\infty}^{\infty}
f(s) g(t-s)\, \rd s$ denotes the convolution. The second function is
\begin{align}
g_2(t) &= \frac{1}{\sqrt{2\pi}} \int_{-\infty}^{\infty} \rd \nu_+\,
e^{i \nu_+ t}
\exp\left( -\frac{\beta^2 \nu_+^2 + 4 \beta \nu_+}{16} \right) \nonumber \\
&= \frac{2\sqrt{2}}{\beta}
\exp\left( \frac{(\beta - 4i t)^2}{4 \beta^2} \right).
\end{align}
For future proofs, it would be useful to provide some upper bounds for these functions.
Using Young’s inequality and taking into account that $|\sin(2t/\beta)|\leq 1$, we get
\begin{align}\label{eq:bound1}
    \int_{-\infty}^\infty & \rd t |g_1(t)| \leq \int_{-\infty}^\infty \rd t \,  \frac{1}{\pi \beta \cosh(\frac{2\pi t}{\beta})}\times \int_{-\infty}^\infty \rd s\frac{\sqrt{2}}{\beta}e^{\frac{1}{4}-\frac{4s^2}{\beta^2}} \nonumber\\
        &= \frac{1}{2\pi}\times \frac{\sqrt{2 \pi}\e^{1/4}}{2}=\frac{e^{1/4}}{2\sqrt{2\pi}}.
\end{align}
Also, we can show that
\begin{align}\label{eq:bound1}
    \int_{-\infty}^\infty \rd t |g_2(t)| = \int_{-\infty}^\infty \rd t \,\frac{2\sqrt{2}}{\beta} e^{\frac{1}{4}}e^{-\frac{4 t^2}{\beta^2}}=e^{\frac{1}{4}}\sqrt{2\pi}.
\end{align}
Next, we prove the following results for the tails of these functions.
\begin{lem}
For any $t_0\geq 0$, the functions $g_1(t)$ and $g_2(t)$ satisfy
\begin{align}\label{eq:tailboundg1}
    &\int_{t_0}^\infty \rd t |g_1(t)|\leq \frac{ \sqrt{2}}{2\pi^{3/2}}\e^{\frac{1}{4}}\left(\e^{\frac{\pi^2}{4}} \e^{-\frac{2\pi t_0}{\beta}}+\frac{1}{2}\erfc\left(\frac{2 t_0}{\beta}\right)\right)\\
\label{eq:tailboundg2}
    &\int_{t_0}^\infty \rd t |g_2(t)|\leq \sqrt{\frac{\pi}{2}} \e^{\frac{1}{4}}\erfc\left(\frac{2t_0}{\beta}\right).
\end{align}
\end{lem}
\begin{proof}
Consider first $g_1(t)$. We have for $t_0>0$
\begin{align}
\begin{split}
&\int_{t_0}^\infty \rd t |g_1(t)|\\
&\leq \int_{t_0}^\infty\rd t\, \int_{-\infty}^\infty \, \rd s  \frac{1}{\pi \beta\cosh\left(\frac{2\pi s}{\beta}\right)}\frac{\sqrt{2}}{\beta}e^{\frac{1}{4}-\frac{4(t-s)^2}{\beta^2}}
\\&\leq\frac{2 \sqrt{2}}{\pi \beta^2}e^{\frac{1}{4}}\int_{t_0}^\infty \rd t \int_{0}^\infty \rd s \,\e^{-\frac{2\pi s}{\beta}}\left(\e^{-\frac{4(t-s)^2}{\beta^2}}+\e^{-\frac{4(t+s)^2}{\beta^2}}\right),
\end{split}
\end{align}
where we used $1/\cosh(x)\leq 2 e^{-|x|}$ in the last line and split the integral over $s$ into two halves $\int_{0}^\infty$ and $\int_{-\infty}^0$.

Let us consider the first integral in the form
$$
I_1 \;=\;\int_{t_0}^{\infty}\!\mathrm{d}t\;\int_{0}^{\infty}\!\mathrm{d}s\;
e^{-\frac{2\pi s}{\beta}}
\,e^{-\frac{4(t-s)^2}{\beta^2}},
$$
one finds by completing the square in $s$ that
$$
\int_{0}^{\infty}\!\mathrm{d}s\;
e^{-\frac{2\pi s}{\beta}}
\,e^{-\frac{4(t-s)^2}{\beta^2}}
=\frac{\beta\sqrt{\pi}}{4}\,e^{\frac{\pi^2}{4}}\,
e^{-\frac{2\pi t}{\beta}}\,
\erfc\Bigl(\frac{\pi}{2}-\frac{2t}{\beta}\Bigr).
$$
Hence, we get
$$
I_1
=\frac{\beta\sqrt{\pi}}{4}\,e^{\frac{\pi^2}{4}}
\int_{t_0}^{\infty}\!\mathrm{d}t\;e^{-\frac{2\pi t}{\beta}}\;
\erfc\Bigl(\frac{\pi}{2}-\frac{2t}{\beta}\Bigr).
$$
Since $\erfc(x)\le2$ for all $x$, we get the bound
$$
I_1
\;\le\; \frac{\beta\sqrt{\pi}}{2}\,e^{\frac{\pi^2}{4}}
\int_{t_0}^\infty\!\mathrm{d}t\;e^{-\frac{2\pi t}{\beta}}
\;=\;
\frac{\beta^2}{4\sqrt{\pi}}\,e^{\frac{\pi^2}{4}}\;
e^{-\frac{2\pi t_0}{\beta}}.
$$
The other term can be bounded as
\begin{align}
\begin{split}
&I_2=\int_{t_0}^\infty\rd t \int_{0}^\infty \rd s \,\e^{-\frac{2\pi s}{\beta}}\e^{-\frac{4(t+s)^2}{\beta^2}}\\&
\leq\int_{t_0}^\infty\rd t \int_{0}^\infty \rd s \,e^{-\frac{2\pi s}{\beta}} e^{-\frac{4t^2}{\beta^2}}= \frac{\beta^2}{8\sqrt{\pi}}\erfc\left(-\frac{2 t_0}{\beta}\right).
\end{split}
\end{align}
This leads us to the expression in Eq.~\eqref{eq:tailboundg1}.
The expression in Eq.~\eqref{eq:tailboundg2} follows from bounding an Gaussian integral.

\end{proof}

\subsection{Zero-Temperature Lindblad operator}
\label{subsec:beta0}

In the infinite-temperature limit, i.e., as $\beta \rightarrow 0$, the function $g(\nu_1, \nu_2)$ in Eq.~\eqref{eqs:g_express} vanishes, implying that $G^{\beta=0}_{a,\alpha} = 0$. Moreover, from Eq.~\eqref{eqs:f_expression}, it follows that 
\begin{align} 
\lim_{\beta \rightarrow 0} f(t) = \delta(t), 
\end{align} 
where $\delta(t)$ denotes the Dirac delta function.

Given this limiting behavior, and choosing the operator set as in Eq.~\eqref{jump_operator_choice}, the Lindblad generator in the zero-temperature case simplifies to 
\begin{align} \mathcal{L}^{\beta = 0}(\rho) &= \sum_{a\in\Lambda}\sum_{\alpha} \Bigl(A^{a,\alpha} \rho A^{a,\alpha\dagger} - \frac{1}{2} \{ A^{a,\alpha\dagger} A^{a,\alpha}, \rho \}\Bigr) \nonumber \\
&= \sum_{a\in\Lambda} \left(\frac{1}{2} \tr_a(\rho)\otimes I_a- \rho\right),
\end{align}
where $\tr_a$ is the partial trace over the single-qubit space 
$\mathsf H_a$ and $I_a\in \mathsf H_a\to \mathsf H_a$ is the identity operator on qubit $a$. This expression corresponds to the generator of the depolarizing noise channel, which drives the system toward the maximally mixed state by uniformly averaging over local Pauli operators~\cite{rouzé2024optimalquantumalgorithmgibbs}.

\section{Bounds on the Trace Distance}\label{sec:tracedistanceerror}

This Appendix gives the proof for Theorem~\ref{thrm:trunc} and Theorem~\ref{thrm:mixingtime}. 

\begin{thrm}[Truncation Error, \cref{thrm:trunc} restated] \label{apx_thrm:trunc}
Let $H$ be a local Hamiltonian on a finite lattice $\Lambda$ with $n$ qubits, $\rho_\beta$ the Gibbs state, and $\mathcal L^{\beta,r}$ the truncated Lindbladian in Eq.~\eqref{eq:trunc_lindbl}. Then, there exists a constant $\beta^*=\frac{1}{500^{D}J}$ such that for any $\beta < \beta^*$,  there exists a unique $\rho_{\beta,r} \in \mathsf M_\Lambda$ satisfying $\mathcal{L}^{\beta,r}(\rho_{\beta,r}) = 0$ and $J>0$, such that
\begin{equation} 
\|\rho_{\beta,r}- \rho_\beta\|_1 = O\left((\beta J)^{r/2}n \log n \right), \end{equation} 
where $\|\cdot\|_1$ denotes the trace norm.
\end{thrm}
and
\begin{thrm}[\cref{thrm:mixingtime} restated] \label{apx_thm:highT1}
Let $H$ be a $(k,l)$-local Hamiltonian with Lieb-Robinson velocity $J>0$ acting on $n$ qubits arranged on a $D$-dimensional lattice, and let $\beta^*=\frac{1}{2020^{D}J}$ . Then, for any $\beta<\beta^*$, state $\rho$ and $\epsilon>0$, the generator $\mathcal L^\beta$ in Eq.~\eqref{eq:LindbladLin} with the choice $q(\nu)=\e^{-\frac{(\beta\nu)^2}{8}}$ and its fixed point $\rho_\beta$ satisfy
\begin{align}\label{equ:mixing_gap}
\|e^{t\mathcal{L}}(\rho)-\rho_\beta\|_1\le \epsilon \|\rho-\rho_\beta \|_1\quad \quad \text{ for all}\quad t=\Omega(\log(n/\epsilon))\,.
\end{align}
From this statement, it follows that for $\beta < \beta^*$, the mixing time of the Lindbladian $\cL^\beta$ scales as $t_{\mathrm{mix}} = \mathcal{O}(\log n)$.
This statement is also true for the truncated Lindbladian $\cL^{\beta,r}$ and its fixed point $\rho_{\beta,r}$, respectively.
\end{thrm}

Note that the results for the nontruncated Lindbladian generators presented in Ref.~\cite{ding2024efficientquantumgibbssamplers} are analogous to those for the generators introduced in Ref.~\cite{chen2023efficient}. In particular, for the latter generators, a logarithmic scaling of the mixing time with system size at inverse temperatures $\beta < \beta^*$ was established in Ref.~\cite{rouzé2024optimalquantumalgorithmgibbs}. By leveraging the structure of their proof, we extend this result to the Lindblad generator defined in Eq.~\eqref{eq:LindbladLin} of Ref.~\cite{ding2024efficientquantumgibbssamplers} for the specific choice
$
q^{a,\alpha}(\nu) \;=\; \exp(-(\beta\nu)^2/8)\,.
$
Crucially, the argument relies only on the quasilocality of the jump operators and hence applies equally to any other choice of $q_{a,\alpha}(\nu)$ for which the corresponding envelope function $f(t)$ in Eq.~\eqref{eq:envelope} decays sufficiently rapidly in~$t$.

Furthermore, quasilocality of the jump operators is preserved under truncation. Consequently, the mixing time remains logarithmic in the system size even when one considers the truncated Lindbladian $\mathcal{L}^{\beta,r}$, albeit now with respect to the perturbed fixed point $\rho_{\beta,r}$. It therefore suffices to prove Theorem~\ref{apx_thm:highT1} for $\mathcal{L}^{\beta,r}$, and hence to establish Theorem~\ref{apx_thrm:trunc}.

The rest of the Appendix section is structured as follows. In \cref{subsec:Preliminaries}, we state the necessary propositions and lemmas to set up the proofs. These preliminaries allow us to construct the proof of Theorem~\ref{apx_thm:highT1} in \cref{subsec:thmrmhighT1}, which in turn relies on Lemma~\ref{lemthmrmhighT1}, proved in \cref{subsec:Tracedistancelemma}. Finally, we complete the present Appendix section by proving Theorem~\ref{apx_thrm:trunc} in \cref{appss:proofonenorm}.

\subsection{Preliminaries}~\label{subsec:Preliminaries}
This subsection collects the necessary lemmas to provide the proofs for \cref{apx_thrm:trunc,apx_thm:highT1}. The proof of Theorem~\ref{apx_thrm:trunc} is built on the following lemma that connects the fixed points of a pair of close Lindbladians:
\begin{lem}\label{fact:mixing_time_to_fixedpoint}
(\cite[Lemma II.1]{chen2023quantum}) Let $\CL_1$ and $\CL_2$ be two generators of Lindbladian evolution with unique fixed points $\rho_{\mathrm{fix}}(\CL_1)$ and $\rho_{\mathrm{fix}}(\CL_2)$, respectively. Then, the trace norm difference between these fixed points satisfies  
\begin{align}
    \big\| \vrho_{\mathrm{fix}}(\CL_1) - \vrho_{\mathrm{fix}}(\CL_2) \big\|_1 
    \leq 4 \, \big\| \CL_1 - \CL_2 \big\|_{1 \to 1} \cdot t_{\mathrm{mix}}(\CL_1),
\end{align}  
where $t_{\mathrm{mix}}(\CL_1)$ denotes the mixing time of the generator $\CL_1$.  
\end{lem}

This lemma establishes a bound relating the distance between the fixed points of two Lindbladian generators, the superoperator difference norm $\bigl\lVert \mathcal{L}_1 - \mathcal{L}_2 \bigr\rVert_{1\to1}$, and the mixing time $t_{\mathrm{mix}}$. The proof can be found in Ref.~\cite{chen2023quantum}.

In order to prove Theorem~\ref{apx_thm:highT1} and a fast mixing time, we make use of the following lemma:
\begin{lem}\label{lemthmrmhighT1}
Consider the oscillator norm
\begin{align*}
\vertiii{X}:=\sum_{a\in\Lambda}\|\delta_a(X)\|_\infty\,,
\end{align*}
where $\delta_a(X):=X-\frac{1}{2}I_a\otimes \operatorname{tr}_a(X)$.
Under the same assumptions as for Theorem~\ref{apx_thm:highT1},
for $\beta< \beta^*$ there exists a $\kappa\equiv \kappa(\beta)$ with $0\leq\kappa<1$ such that, for any initial observable $X$, 
\begin{align}\label{gradientest}
\vertiii{\e^{t\mathcal{L}^{\beta \dagger}}(X)}\le \e^{-(1-\kappa)t}\vertiii{X}\,.
\end{align}
Additionally, for $\beta<\beta^*$, there exists a 
$\kappa'(\beta)$ with $0\leq\kappa'<1$ such that, for any initial observable $X$, 
\begin{align}\label{gradientesttruncated}
\vertiii{\e^{t\mathcal{L}^{\beta, r \dagger}}(X)}\le \e^{-(1-\kappa')t}\vertiii{X}\,.
\end{align}
\end{lem}
We provide the proof of this lemma, which relies on the quasilocality of the Lindbladian, in \cref{subsec:Tracedistancelemma}.
It follows the steps of Ref.~\cite{rouzé2024optimalquantumalgorithmgibbs}.
Furthermore, the proof gives as a by-product a bound on $\norm{\cL^\beta-\cL^{\beta,r}}_{1\rightarrow 1}$ which we use in  combination with \cref{fact:mixing_time_to_fixedpoint} and \cref{apx_thm:highT1} to obtain the proof of \cref{apx_thrm:trunc}.

\subsection{Proof of Theorem~\ref{apx_thm:highT1}}\label{subsec:thmrmhighT1}

The proof follows the same steps as the proof for the logarithmic mixing time in Ref.~\cite{rouzé2024optimalquantumalgorithmgibbs} for the specific instance of quantum Gibbs sampler in Ref.~\cite{chen2023efficient}. 

The uniqueness of the fixed point follows because the operators $A^{a,\alpha}$ in the Lindbladian in Eq.~\eqref{eq:trunc_lindbl} form a complete set of generators and thus the Lindbladian is irreducible~\cite{ding2024efficientquantumgibbssamplers,rouze2024efficientthermalizationuniversalquantum}.
Consider
\begin{align}\label{eq:boundtimes2}
\|\e^{t\mathcal{L}^\beta}(\rho)-\rho_\beta\|_1&=\sup_{\|X\|_\infty\le 1}\tr(X(e^{t\mathcal{L}^\beta}(\rho)-\rho_\beta))=
\sup_{\|X\|_\infty\le 1}\tr(\e^{t\mathcal{L}^{\beta\dagger}} (X)\left(\rho-\rho_\beta\right) ).
\end{align}
Here we used the definition of the adjoint operator $\tr(X \e^{t\mathcal{L}^\beta}(\rho))=\tr(\e^{t\mathcal{L}^{\beta\dagger}}(X)\rho)$, and the invariance of the steady state under dissipative evolution $e^{t\mathcal{L}^\beta}(\rho_\beta)=\rho_\beta$.
Observe now, that $\tr(\rho-\rho_\beta)=0$ and thus 
\begin{align}
\begin{split}
    &\tr(e^{t\mathcal{L}^{\beta\dagger}} (X)\left(\rho-\rho_\beta\right))\\
    =&\tr(e^{t\mathcal{L}^{\beta\dagger}} (X)\left(\rho-\rho_\beta\right))-2^{-n}\tr(e^{t\mathcal{L}^{\beta\dagger}} (X))\tr(\rho-\rho_\beta)\\=&
    \tr(e^{t\mathcal{L}^{\beta\dagger}} (X)\left(\rho-\rho_\beta\right))-\tr(2^{-n}\tr(e^{t\mathcal{L}^{\beta\dagger}} (X)) I\left(\rho-\rho_\beta\right))\\
    =&\tr(\left(e^{t\mathcal{L}^{\beta\dagger}} (X)-2^{-n}\tr(e^{t\mathcal{L}^{\beta\dagger}} (X)) I\right )\left(\rho-\rho_\beta\right)).
\end{split}
\end{align}
We can use this expression in Eq.~\eqref{eq:boundtimes2} and obtain together with Lemma~\ref{lemthmrmhighT1} that
\begin{align}\label{eq:boundtimes}
\begin{split}
\|e^{t\mathcal{L}^\beta}(\rho)-\rho_\beta\|_1&=\sup_{\|X\|_\infty\le 1}\|e^{t\mathcal{L}^{\beta\dagger}} (X)-2^{-n}\tr(e^{t\mathcal{L}^{\beta\dagger}} (X))I\|_\infty\|\rho-\rho_\beta \|_1\\
&\le \sup_{\|X\|_\infty\le 1} \vertiii{e^{t\mathcal{L}^{\beta\dagger}} (X)}\,\|\rho-\rho_\beta \|_1 \  \\
&\le e^{-(1-\kappa)t}\sup_{\|X\|_\infty\le 1}\vertiii{X}\, \|\rho-\rho_\beta \|_1\\
&\le 2n \|\rho-\rho_\beta\|_1 e^{-(1-\kappa)t}\,,
\end{split}
\end{align}
In the first line, we have used the fact that the trace over many sites can be written sequentially as a composition of trace over single sites. 
In the last line, we used
\begin{align}
\|\delta_a(X)\|_\infty \leq \|X\|_\infty + \frac{1}{2} \| I_a\otimes \operatorname{tr}_a(X)\|_\infty \leq 2 \|X\|_\infty, 
\end{align}
since the operator norm is not increased by partial trace, and thus $\vertiii{X}\leq 2n$.
From this it follows that
\begin{align}
\|e^{t\mathcal{L}}(\rho)-\rho_\beta\|_1\le \epsilon \|\rho-\rho_\beta \|_1
\end{align}
with 
\begin{align}
    t=\frac{1}{1-\kappa}\log\left(\frac{2 n}{\epsilon}\right).
\end{align}
This concludes the proof of Theorem~\ref{apx_thm:highT1}. 
\medskip
\subsection{Proof of Lemma~\ref{lemthmrmhighT1} }\label{subsec:Tracedistancelemma}

In this subsection, we give the proof of Lemma~\ref{lemthmrmhighT1}, following the same steps as Ref.~\cite{rouzé2024optimalquantumalgorithmgibbs}.

\begin{proof}[Proof of Lemma~\ref{lemthmrmhighT1}:]
We define for any $r_0\in\mathbb{N}$ and $\beta$ the quantities
\begin{align}\label{eq:defDelta}
\Delta(r_0)&=\sum_{r\ge r_0}\|\mathcal{L}^{\beta,r\dagger}\!\!-\mathcal{L}^{\beta,(r-1)\dagger }\|_{\infty\to\infty}\!,\\\label{eq:defeta}
\eta(\beta)&=\|\mathcal{L}^{\beta\dagger}-\mathcal{L}^{0 \dagger }\|_{\infty\to\infty},
\end{align}

Each term appearing in the oscillator norm can be bounded locally~\cite{rouzé2024optimalquantumalgorithmgibbs,temme2015faststabilizerhamiltoniansthermalize}. These local bounds allow bounding the decay of the oscillator norm, as is shown in~\cite{temme2015faststabilizerhamiltoniansthermalize}. More specifically, consider the following quantity \cite[Eq.~(A35)]{rouzé2024optimalquantumalgorithmgibbs}, which arises from geometrical considerations of the $D$-dimensional lattice and the behavior of the oscillator norm:

\begin{align}\label{eq:kappabound}
k(r_0,\beta)= 4(2r_0+1)^{2D}\,\eta(\beta)+u(r_0),
\end{align}
where $u(r_0)$ is defined as 
\begin{align}\label{eq:deff0}
\begin{split}
    u(r_0)=5  (2r_0+1)^{2D} \Delta(r_0) +\left( 5+2r_0+2(2r_0+1)^{D} \right)\sum_{l \ge r_0} (2l+1)^{2D-1}\Delta(l) + 2 \sum_{l'\ge r_0} \sum_{l \ge l'} (2l+1)^{2D-2} \Delta(l).
\end{split}
\end{align}
Concretely, if there exists an $r_0\in\mathbb{N}$ and $\beta$ such that $k(r_0,\beta)<1$, then the inequality Eq.~\eqref{gradientest} is satisfied by the choice $\kappa=k(r_0,\beta)$.
The remaining strategy for the proof is thus the following:
First, we provide bounds on $\Delta(r)$ and $\eta(\beta)$ for the Lindbladian $\cL^\beta$.
This allows us to bound $u(r_0)$ for $\beta<\beta^*$, and a properly chosen $r_0$ such that $k(r_0,\beta)<1$, which will complete the proof.

\subsubsection{Bounding $\Delta(r)$}\label{appssbounding}

In the following, the steps in [Ref.~\cite{rouzé2024optimalquantumalgorithmgibbs}, Appendix B] are adapted to the Lindbladian $\mathcal{L}^\beta$ in Eq.~\eqref{eq:LindbladLin}.

As the first step, the Lindbladian 
Eq.~\eqref{eq:LindbladLin} can be rewritten as a telescopic sum
\begin{align}\label{L_telescopic}\mathcal{L}^\beta_{a,\alpha}=\mathcal{L}^{\beta,0}_{a,\alpha}+\sum_{r=0}^{\infty}\left(\mathcal{L}^{\beta,r+1}_{a,\alpha}-\mathcal{L}^{\beta,r}_{a,\alpha}\right)\equiv\mathcal{L}^{\beta,0}_{a,\alpha} + \sum_{r=0}^\infty \mathcal{E}^{\beta,r}_{a,\alpha},
\end{align}
with $\mathcal{E}^{\beta,r}_{a,\alpha}:=\mathcal{L}^{\beta,r+1}_{a,\alpha}-\mathcal{L}^{\beta,r}_{a,\alpha}$ the finite difference operator supported on $B_a(r+1)$. \footnote{When $r > \max_{b \in \Lambda} \ell(a,b)$ we define the error term to be zero.}

In this representation, the difference between $\mathcal{L}^{\beta}_{a,\alpha}$ and $\mathcal{L}^{\beta,r}_{a,\alpha}$  is given by
\begin{align}\label{eq:error:truncation}
    \mathcal{L}^\beta_{a,\alpha}-\mathcal{L}^{\beta,r}_{a,\alpha}=\sum_{r'=r}^{\infty} \mathcal{E}^{\beta,r'}_{a,\alpha} 
\end{align}
It remains to be shown that the terms $ \mathcal{E}^{\beta,r}_{a,\alpha}$ decay sufficiently fast with the radius $r$. 

We aim to show that there is a sequence of Lindbladians with jump operators with support on a region of radius $r$ such that, for a fast-decaying function $\zeta(r)$, 
\begin{align}\label{eq:defzeta}
&\normp{ \mathcal{E}^{\beta,r}_{a,\alpha}}{\infty\rightarrow \infty} \le \zeta(r)\,.
\end{align} 
Recall the jump operators and coherent term of the Lindbladian $\mathcal{L}^{\beta,r}_{a,\alpha}$:
\begin{align}\label{eq:jumplocaltrunc}
L^{\beta,r}_{a,\alpha}:=\int_{-\infty}^{\infty} e^{iH_rt}A^{a,\alpha} e^{-iH_r t}\,f(t)\,\mathrm{d}t\,,
\end{align}
and coherent part
\begin{align}\label{eq:coherentlocaljump}
G^{\beta,r}_{a,\alpha} &\equiv \int_{-\infty}^\infty g_1(t) e^{-i H_r t} \int_{-\infty}^\infty g_2(t') e^{i H_r t'}A^{a,\alpha\dag} e^{-2i H_r t'}A^{a,\alpha} e^{i H_r t'}\rd t' e^{i H_r t}\rd t \\ & = \int_{-\infty}^\infty  \int_{-\infty}^\infty g_1(t) g_2(t') A_r^{a,\alpha \dag} (t-t') A_r^{a,\alpha}(t+t') \rd t' \rd t\,.
\end{align}
The explicit expressions for $f(t)$, $g_1(t)$ and $g_2(t)$ with the choice of $q(\nu)$ in Eq.~\eqref{eq:Gauss} are given by~[cf. \cref{subsec:TimreprLindbladians}]:
\begin{align}\label{eq:timejump}
    f(t)&=\sqrt{\frac{2}{\pi \beta^2}}e^{\frac{(\beta-4 i t)^2}{8\beta^2}},\\
     g_1(t)&=\left[\frac{-1}{\pi \beta \cosh(\frac{2\pi t}{\beta})}\right]*_t\left[\frac{\sqrt{2}}{\beta}e^{\frac{1}{4}-\frac{4t^2}{\beta^2}}\sin\left(\frac{2 t}{\beta}\right)\right],\\
 g_2(t) &=\frac{2\sqrt{2}}{\beta} e^{\frac{(\beta-4 it)^2}{4 \beta^2}}.
\end{align}

In the following, we consider the dissipative part of the Lindbladian $\mathcal{D}_{a,\alpha}^{\beta,r}$  defined as 
\begin{align}\label{Eq:triangle}
    \forall X: \qquad \mathcal{D}_{a,\alpha}^{\beta,r}(X):=L_{a,\alpha}^{\beta,r} X L_{a,\alpha}^{\beta,r\dagger}-\frac{1}{2}\{L_{a,\alpha}^{\beta,r\dagger} L_{a,\alpha}^{\beta,r},X\}.
\end{align}
To proceed, we note that using the triangle inequality, we have
\begin{align}\label{eq:part1}
\begin{split}
& \norm{L_{a,\alpha}^{\beta,(r+1)}X L_{a,\alpha}^{\beta,(r+1) \dagger}-L_{a,\alpha}^{\beta,r}X L_{a,\alpha}^{\beta,r\dagger}}_{\infty}=
\norm{(L_{a,\alpha}^{\beta,(r+1)}-L_{a,\alpha}^{\beta,r})X L_{a,\alpha}^{\beta,(r+1)\dagger}-L_{a,\alpha}^{\beta,r}X (L_{a,\alpha}^{\beta, r\dagger}-L_{a,\alpha}^{\beta,(r+1)\dagger})}_{\infty}\\
&
\leq \norm{(L_{a,\alpha}^{\beta,(r+1)}-L_{a,\alpha}^{\beta,r})X L_{a,\alpha}^{\beta,(r+1)\dagger}}_\infty+\norm{L_{a,\alpha}^{\beta,r}X( L_{a,\alpha}^{\beta,r\dagger}-L_{a,\alpha}^{\beta,(r+1)\dagger})}_\infty
\end{split}
\end{align}
We can further bound the terms using Hoelder's inequality:
\begin{align}\label{eq:part2}
\begin{split}
&\norm{(L_{a,\alpha}^{\beta,(r+1)}-L_{a,\alpha}^{\beta,r})X L_{a,\alpha}^{\beta,(r+1)\dagger}}_\infty\leq \norm{L_{a,\alpha}^{\beta,(r+1)}-L_{a,\alpha}^{\beta,r}}_\infty\norm{X}_\infty\norm{L_{a,\alpha}^{\beta,(r+1)\dagger}}_\infty
\end{split}
\end{align}
Using the same steps, the same result is obtained for  $\frac{1}{2}\{L_{a,\alpha}^{\beta,r\dagger} L_{a,\alpha}^{\beta,r},X\} $. Combining Eq.~\eqref{eq:part1} and Eq.~\eqref{eq:part2}, this gives
\begin{align}\label{eq:part3D}
\begin{split}
    &\norm{\mathcal{D}_{a,\alpha}^{\beta,(r+1)}-\mathcal{D}_{a,\alpha}^{\beta,r}}_{\infty\rightarrow \infty}=\sup_X \frac{\norm{\mathcal{D}_{a,\alpha}^{\beta,(r+1)}(X)-\mathcal{D}_{a,\alpha}^{\beta,r}(X)}_\infty}{\norm {X}_\infty}
    \leq 2 \left(\norm{L_{a,\alpha}^{\beta,r}}_\infty +\norm{L_{a,\alpha}^{\beta,(r+1)}}_\infty\right)\norm{L_{a,\alpha}^{\beta,(r+1)}-L_{a,\alpha}^{\beta,r}}_\infty
\end{split}
\end{align}
Consider the last factor:
\begin{align}
\norm{L_{a,\alpha}^{\beta,(r+1)}-L_{a,\alpha}^{\beta,r}}_\infty= \norm{\int_{-\infty}^{\infty} \mathrm{d}t f(t) e^{-iH_{r+1}t}A^{a,\alpha} e^{iH_{r+1}t}-e^{-iH_{r}t}A^{a,\alpha} e^{iH_rt}}_\infty.
\end{align}

In order to bound the integral, we split the integration range in a short time interval $t\in [-t_0,t_0]$ and long-time tails $|t|>t_0$.
The choice of $t_0$ will be fixed below such that the resulting bound shows explicit decay with $r$ and it is explicit that it vanishes in the limit $\beta\rightarrow 0$.

For the short time interval $t\in [-t_0,t_0]$, we use the
Lieb-Robinson bound Eq.~\eqref{eq:LiebRobinson}. For the long-time tails, we use the trivial upper bound 
\begin{align*}
\norm{e^{-iH_{r+1}t}A^{a,\alpha} e^{iH_{r+1}t}-e^{-iH_{r}t}A^{a,\alpha} e^{iH_rt}}_\infty \leq \norm{e^{-iH_{r+1}t}A^{a,\alpha} e^{iH_{r+1}t}}_\infty+\norm{e^{-iH_{r}t}A^{a,\alpha} e^{iH_rt}}_\infty\leq 2 \norm{A^{a,\alpha}}_\infty.
\end{align*}
The last factor in Eq.~\eqref{eq:part3D} can then be bounded by
\begin{align}
\norm{L_{a,\alpha}^{\beta,(r+1)} - L_{a,\alpha}^{\beta,r}}_\infty &\le  \norm{A^{a,\alpha}}_\infty \int_{-t_0}^{t_0} |f(t)| \frac{(2 J \vert t \vert )^r}{r!} \rd t +  4\norm{A^{a,\alpha}}_\infty \int_{t_0}^\infty  f(t) \rd t\nonumber
\\
&\le \norm{A^{a,\alpha}}_\infty \Biggl(\frac{(2 J \e \vert t_0 \vert )^r}{ r^r} \int_{-\infty}^{\infty} |f(t)|\rd t+ 4 \int_{t_0}^\infty  |f(t)| \rd t\nonumber  \Biggl)\\
& \le \norm{A^{a,\alpha}}_\infty e^{\frac{1}{8}}\Biggl(\frac{(2 J \e \vert t_0 \vert )^r}{ r^r} + 2 e^{\frac{-2 t_0^2}{\beta^2}} \frac{\beta}{\sqrt{2\pi}t_0}\Biggl).
\end{align}
In the second line, we used the estimate $r! \ge \frac{r^r}{e^{r-1}}$ and in the last line the upper bound on the complementary error function $\erfc(x) \le \frac{e^{-x^2}}{\sqrt{\pi}x}$ to evaluate the integrals.

With the choice $t_0 = \frac{r \gamma(\beta J)} {2J e}$ with $\gamma(x)=\frac{\sqrt{x}}{1+\sqrt{x}}$, this gives 
\begin{align}\label{eq:abovegamma}
\norm{L_{a,\alpha}^{\beta,(r+1)} - L_{a,\alpha}^{\beta,r}}_\infty &\le e^{\frac{1}{8}} \norm{A^{a,\alpha}}_\infty \left ( \gamma(\beta J)^r +  \frac{  4 \beta J e }{\sqrt{2\pi} r \gamma(\beta J)} \e^{-\frac{r^2 \gamma(\beta J)^2}{2(\beta J \e )^2}} \right)
\end{align}
Furthermore, we have 
\begin{align}
\norm{L_{a,\alpha}^{\beta,(r+1)}}_\infty\leq e^{1/8}\norm{A^{a,\alpha}}_\infty.
\end{align}
Thus we have 
\begin{align}\label{eq:defxi1}
   \norm{  \mathcal{D}_{a,\alpha}^{\beta,(r+1)}-\mathcal{D}_{a,\alpha}^{\beta,r} }_{\infty\to \infty} \le  4 \norm{A^{a,\alpha}}^2_\infty e^{\frac{1}{4}} \left ( \gamma(\beta J)^r +  \frac{  4 \beta J e }{\sqrt{2\pi} r \gamma(\beta J)} \e^{-\frac{r^2 \gamma(\beta J)^2}{2(\beta J \e )^2}} \right)\equiv\norm{A^{a,\alpha}}_\infty^2  \xi_1(r,\beta J)\,.
\end{align}

For the coherent part, similar to our derivation of \cref{eq:part1,eq:part2}, it can be shown that
\begin{align}
\norm{ -i[G_{a,\alpha}^{\beta,r}-G_{a,\alpha}^{\beta,(r+1)},\cdot] }_{\infty\rightarrow\infty}\leq 2 \norm{G_{a,\alpha}^{\beta,r}-G_{a,\alpha}^{\beta,(r+1)} }_{\infty}
\end{align}
and
\begin{align}
  \norm{  G_{a,\alpha}^{\beta,r}-G_{a,\alpha}^{\beta,(r+1)} }_\infty \le \int_{-\infty}^\infty  \int_{-\infty}^\infty \vert g_1(t) g_2(t') \vert \norm{ A_r^{a,\alpha} (t-t')^{\dag} A_r^{a,\alpha}(t+t') -  A_{r+1}^{a,\alpha\dag} (t-t') A_{r+1}^{a,\alpha}(t+t') }_\infty \rd t' \rd t . \label{eq:gnormdifference}
\end{align}
We can split the interval into five regions as follows, with a different cutoff $t_0$:
\begin{equation}
 \int_{-\infty}^\infty  \int_{-\infty}^\infty \dots = \int_{-t_0}^{t_0}  \int_{-t_0}^{t_0} + \int_{t_0}^\infty  \int_{-t_0}^{t_0} + \int_{-\infty}^\infty\int_{t_0}^\infty + \int_{-\infty}^{-t_0} \int_{-t_0}^{t_0} +\int_{-\infty}^\infty \int_{-\infty}^{-t_0}.
\end{equation}
The integral can be bounded as:
\be \int_{-\infty}^\infty  \int_{-\infty}^\infty |\dots|\leq \int_{-t_0}^{t_0}  \int_{-t_0}^{t_0} + \int_{t_0}^\infty  \int_{-\infty}^\infty + \int_{-\infty}^\infty\int_{t_0}^\infty + \int_{-\infty}^{-t_0} \int_{-\infty}^\infty +\int_{-\infty}^\infty \int_{-\infty}^{-t_0}.
\ee
Consider the first integral.
With a similar argument as \cref{eq:part1,eq:part2,eq:part3D}, we obtain
\begin{align}
 \int_{-t_0}^{t_0}  \int_{-t_0}^{t_0}   \vert g_1(t) g_2(t') \vert \norm{ A_r^{a,\alpha\dag} (t-t') A_r^{a,\alpha}(t+t') -  A_{r+1}^{a,\alpha\dag} (t-t') A_{r+1}^{a,\alpha}(t+t') }_\infty \rd t' \rd t \le e^{1/2} \frac{(4 J   e \vert t_0 \vert )^r}{2r^{r}} \times \nonumber  \\ \norm{A^{a,\alpha \dag} A^{a,\alpha}}_\infty,
\end{align}
where we used the Lieb-Robinson bound and that $\int_{-\infty}^\infty  \int_{-\infty}^\infty \vert g_1(t) g_2(t') \vert \rd t' \rd t \le e^{1/2}/2$ [cf. \cref{subsec:TimreprLindbladians}] 

For the other long-time contributions, we simply bound
\be
\norm{ A_r^{a,\alpha \dag} (t-t') A_r^{a,\alpha}(t+t') -  A_{r+1}^{a,\alpha \dag} (t-t') A_{r+1}^{a,\alpha}(t+t') }_\infty \le 2 \norm{A^{a,\alpha\dag} A^{a,\alpha}}_\infty = 2 \norm{ A^{a,\alpha}}_\infty^2.
\ee
Given the tail bounds [cf Eq.~\eqref{eq:tailboundg1} and Eq.~\eqref{eq:tailboundg2}]
\begin{align}
    &\int_{t_0}^\infty \vert g_1(t) \vert \text{d}t \le  \frac{\sqrt{2}}{2 \pi^{3/2}}\e^{\frac{1}{4}}\left(\e^{\frac{\pi^2}{4}} \e^{-\frac{2\pi t_0}{\beta}}+\frac{1}{2}\erfc\left(\frac{2 t_0}{\beta}\right)\right) \\
    & \int_{t_0}^\infty \vert g_2(t) \vert \text{d}t \le \sqrt{\frac{\pi}{2}} \e^{\frac{1}{4}}\erfc\left(\frac{2t_0}{\beta}\right), 
\end{align}
the contribution of all the other terms in the integral in \cref{eq:gnormdifference} can be bounded as
\begin{align}
2\left(  \int_{-\infty}^\infty\int_{t_0}^\infty +... \right)\vert g_1(t) g_2(t') \vert \rd t' \rd t \norm{A^{a,\alpha}}_\infty^2 \le 2\left(\frac{2 e^{\frac{1}{2}}}{\pi}\e^{\frac{\pi^2}{4}} \e^{-\frac{2\pi t_0}{\beta}}+ e^{\frac{1}{2}}\frac{\beta \e^{-\frac{2 t_0^2}{\beta^2}}}{2 t_0\sqrt{\pi}}\right) \norm{A^{a,\alpha}}_\infty^2.
\end{align}
With the choice $t_0 = \gamma(\beta J) r/ (4 J \e)   $ we obtain

\begin{align}\label{eq:defxi2}
  \norm{ -i[G_{a,\alpha}^{\beta,r}-G_{a,\alpha}^{\beta,(r+1)},\cdot] }_{\infty\rightarrow\infty} & \le \e^{1/2} \norm{A^{a,\alpha}}_\infty^2 \left ( \gamma(\beta J)^r +\frac{8}{\pi} \e^{\frac{\pi^2}{4}- \frac{\pi \gamma(\beta J)}{2 \beta J \e } r  }+\frac{8 \beta J \e  }{\sqrt{\pi} r \gamma(\beta J) } \e^{- \frac{(r \gamma(\beta J))^2}{8 ( \beta J  \e )^2}}  \right) \\& \equiv \norm{A^{a,\alpha}}_\infty^2 \xi_2(r,\beta J).
\end{align}

With the assumption $\norm{A^{a,\alpha}}_\infty \le 1$, we can take  $ \zeta(r)$ in Eq.~\eqref{eq:defzeta} to be
\begin{align}
    \zeta(r) = \xi_1(r,\beta J) +\xi_2(r,\beta J)\,.
\end{align}
$\zeta(r)$ is exponentially decaying in $r$ and vanishes as $\beta \rightarrow 0$.
With the definition $\upsilon:=\beta J/\gamma(\beta J)=\sqrt{\beta J}(1+\sqrt{\beta J})$ and taking $r\ge 1$, a crude bound for $\zeta(r)$ is given by
\begin{align} \label{eq:3terms}
\zeta(r)&\le 7 \gamma(\beta J)^r+44\upsilon e^{-\frac{r^2 }{2 e^2 \upsilon^2}}+50 e^{-\frac{\pi r}{2 e \upsilon }}.
\end{align}
\begin{figure}[h!]
	\centering
	\includegraphics[width=0.5\textwidth]{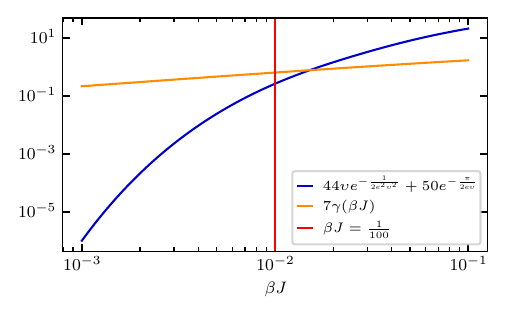}
    \caption{The contributions in Eq.~\eqref{eq:3terms} for $r=1$. As long $\beta J\le\frac{1}{100}$, the term $7\gamma(\beta J)^r$ is dominant.} 
	\label{fig:Bounds}
\end{figure}
For $\beta J\le 1/100$, the term $7 \gamma(\beta J)^r$ in Eq.~\eqref{eq:3terms} is dominant\footnote{The first term decreases slower with increasing $r$ than the other two terms, therefore it suffices to check $r=1$.}, therefore the other terms contribute at most $7\gamma(\beta J)^r$. Thus for $\beta J\le\frac{1}{100}$, we obtain
\begin{align}
\norm{\mathcal{E}_{a,\alpha}^{\beta,r}}_{\infty\rightarrow \infty}\leq \zeta(r)&\le 14 \gamma(\beta J)^r,
\end{align}
and therefore
\begin{align}
\Delta(r_0)=\sum_{r=r_0}^\infty \norm{\mathcal{E}_{a,\alpha}^{\beta,r}}_{\infty\rightarrow \infty}=\sum_{r\ge r_0}\zeta(r)&\le \frac{14 \gamma(\beta J)^{r_0}}{1-\gamma(\beta J)}.
\end{align}
The proper $r_0$ will be chosen below  in \cref{subsec:boundingu(r_0)}.
\subsubsection{Bounding $\eta(\beta)$}

As a next step, we have to bound  $\eta(\beta) = \|\mathcal{L}^{\beta\dagger}-\mathcal{L}^{0 \dagger }\|_{\infty\to\infty}$.

Consider $\mathcal{L}^{0}$. It can be shown [cf. \cref{subsec:beta0}] that
$\mathcal{L}^{0\dagger}(X):= \sum_{a\in \Lambda}\left( \frac{1}{2}I_a\otimes\operatorname{tr}_a(X)-X\right)$ for any $X$. This is the generator of the fully depolarizing channel. 

As before, we consider coherent and dissipative parts separately. For the jump operators, we have
\begin{align}
\norm{L_{a,\alpha}^{\beta}-L_{a,\alpha}^{0}}_\infty &\le \int_{-\infty}^{\infty} \norm{A^{a,\alpha}(t)-A^{a,\alpha} }_\infty \,f(t)\,\rd t\qquad 
\\ & \le \norm{[H,A]}_\infty \int_{-\infty}^{\infty}  \vert t   f(t)\vert \text{d}t
\\ & \le \frac{2}{\sqrt{2\pi}}e^{\frac{1}{8}} \beta J \norm{A^{a,\alpha}}_\infty.
\end{align}
Here we have used $\norm{A(t)-A}_\infty\le \vert t \vert \norm{[H,A]}_\infty\le 2J\vert t \vert \norm{A}_\infty $~\cite{rouzé2024optimalquantumalgorithmgibbs}. With $\norm{A^{a,\alpha}}_\infty \le 1$, we obtain, with a similar derivation as \cref{eq:part3D}, 
\begin{align}\label{eq:etaboundA}
\|\mathcal{D}_{a,\alpha}^{\beta\dagger}-\mathcal{D}_{a,\alpha}^{0 \dagger }\|_{\infty\to\infty}\le  \frac{8 \beta J}{\sqrt{2 \pi}}e^{\frac{1}{8}}.
\end{align}

For the coherent part, we introduce $c_1 =\int_{-\infty}^\infty  \int_{-\infty}^\infty  g_1(t) g_2(t')  \rd t' \rd t$.
With our choice of the jump operators $A^{a,\alpha}$, we have $(A^{a,\alpha})^2=\mathbb{I}$, and

\begin{align}
\begin{split}
 \norm{  G^\beta_{a,\alpha}- c_1 I }_\infty &\le \int_{-\infty}^\infty  \int_{-\infty}^\infty \vert g_1(t) g_2(t') \vert {\norm{A^{a,\alpha} (-t') A^{a,\alpha}(t')- I }_\infty}\,\rd t' \rd t
\\ &
=\int_{-\infty}^\infty  \int_{-\infty}^\infty \vert g_1(t) g_2(t') \vert {\norm{\left(A^{a,\alpha} (-t')-A^{a,\alpha}\right) A^{a,\alpha} (t')+ A^{a,\alpha}\left(A^{a,\alpha}(t')- A^{a,\alpha}\right) }_\infty}\,\rd t' \rd t\\&
\le 4  J \norm{A^{a,\alpha}}_\infty^2 \int_{-\infty}^\infty  \int_{-\infty}^\infty \vert g_1(t) g_2(t') \vert \vert t'\vert \rd t' \rd t = \frac{e^{\frac{1}{2}}}{\sqrt{\pi}} \beta J \norm{A^{a,\alpha}}_\infty^2.
\end{split}
\end{align}
This gives
\begin{align}\label{eq:etaboundG}
\norm{ -i[ G^\beta_{a,\alpha}- c_1 I,\cdot] }_{\infty\rightarrow\infty}\leq \frac{2 e^{\frac{1}{2}}}{\sqrt{\pi}} \beta J \norm{A^{a,\alpha}}_\infty^2.
\end{align}
Putting Eq.~\eqref{eq:etaboundA} and Eq.~\eqref{eq:etaboundG} together, this means we can bound
\begin{align}
\eta(\beta) =  \|\mathcal{L}^{\beta\dagger}-\mathcal{L}^{0 \dagger }\|_{\infty\to\infty} \leq \|\mathcal{D}_{a,\alpha}^{\beta\dagger}-\mathcal{D}_{a,\alpha}^{0 \dagger }\|_{\infty\to\infty} + \norm{ -i[ G^\beta_{a,\alpha}- c_1 I,\cdot] }_{\infty\rightarrow\infty} \leq \left( \frac{2e^{1/2}}{\sqrt{\pi}} + \frac{8e^{1/8}}{\sqrt{2\pi}}\right) \beta J <  6 \beta J.
\end{align}
\subsubsection{Bounding $u(r_0)$}\label{subsec:boundingu(r_0)}
It remains to get a bound for $u(r_0)$ in Eq.~\eqref{eq:deff0}:
\begin{align}\label{eq:deff0again}
\begin{split}
    u(r_0)=5  (2r_0+1)^{2D} \Delta(r_0) +\left( 5+2r_0+2(2r_0+1)^{D} \right)\sum_{l \ge r_0} (2l+1)^{2D-1}\Delta(l) + 2 \sum_{l'\ge r_0} \sum_{l \ge l'} (2l+1)^{2D-2} \Delta(l).
\end{split}
\end{align}
We define $s=(\beta J)^{1/D}$. By definition $\gamma(\beta J)<(\beta J)^{1/2} = s^{D/2}$. For $r_0\geq 1$, and $\beta J\leq \frac{1}{100}$, 
$\gamma(\beta J)<\frac{1}{2}$ and thus $\Delta(r)<28 \gamma(\beta J)$.
We can thus bound
\begin{align}
    &\sum_{l \ge r_0} (2l+1)^{2D-1}\Delta(l)\leq \sum_{l \ge r_0} (2l+1)^{2D}\Delta(l) \leq \frac{14}{1-\gamma(\beta J)} \sum_{l \ge r_0} (2l+1)^{2D} \gamma(\beta J)^l\leq 28 \sum_{l\geq r_0} \left((2l+1)^4 s^l\right)^{D/2}\nonumber\\&\leq 28 \left(\sum_{l\geq r_0} (2l+1)^2 s^{l/2}\right)^{D}.
\end{align}
Similarly, we can bound
\begin{align}
\begin{split}
    &2 \sum_{l'\ge r_0} \sum_{l \ge l'} (2l+1)^{2D-2}\Delta(l)\leq \sum_{l\ge r_0} (2l+1)^{2D}\Delta(l)\leq 28 \left(\sum_{l\geq r_0} (2l+1)^2 s^{l/2}\right)^{D}.
\end{split}
\end{align}

We have thus
\begin{align}
    u(r_0)&\leq 140((2r_0+1)s^{r_0/4})^{2D}+28 \left( 6+2r_0+2(2r_0+1)^{D} \right)\left(\sum_{l\geq r_0} (2l+1)^2 s^{l/2}\right)^{D}\nonumber\\&\leq 140((2r_0+1)^2 s^{r_0/2})^{D}+140(2r_0+1)^{D} \left(\sum_{l\geq r_0} (2l+1)^2 s^{l/2}\right)^{D}.
\end{align}
For $s<\frac{1}{2020}$ and $r_0=4$, it can be easily shown that $k(r_0,\beta)<1$ in Eq.~\eqref{eq:kappabound}, and this concludes the proof of \cref{gradientest}, with $\beta J<\frac{1}{2020^D}$.

Finally, we prove the statement for truncated Lindbladians (\cref{gradientesttruncated}).
For truncated Lindbladians $\mathcal{L}^{\beta,r}$ we have that $\Delta(l) = 0$ for $l > r$.
Similarly to Eq.~\eqref{eq:kappabound}, the exponential decay of the oscillator norm is characterized by 
\begin{align}\label{eq:kappaboundprime}
k'(r_0,\beta)= 4(2r_0+1)^{2D}\,\eta'(\beta)+u'(r_0).
\end{align}
Here $u'(r)$ is a modified version of $u(r)$, with
\begin{align}
 u'(r_0) &=5  (2r+1)^{2D} \Delta(r_0) +\left( 5+2r+2(2r_0+1)^{D} \right)\sum_{l \ge r_0}^r (2l+1)^{2D-1}\Delta(l) + 2 \sum_{l'\ge r} \sum_{l \ge l'}^r (2l+1)^{2D-2} \Delta(l) \\
&\leq 140((2r_0+1)^2 s^{r_0/2})^{D}+140(2r_0+1)^{D} \left(\sum_{l\geq r_0} (2l+1)^2 s^{l/2}\right)^{D},
\end{align}
and $\eta'(\beta)=\eta(\beta)$.
For $r_0>r$, $u'(r_0)=0$ and $k(r_0,\beta)$ in Eq.~\eqref{eq:kappabound} is bound by 
 $\eta'(\beta) \leq 6\beta J = 6s^D$.
 We see thus that for the truncated case as well, for $r\geq 1$, $\kappa'=k'(r,\beta)$ in \cref{eq:kappabound} is less than 1 for $s<1/2020$.
This proves \cref{gradientesttruncated} as well with $\beta J < \frac{1}{2020^D}$, concluding the proof.
\end{proof}

\medskip
\subsection{Proof of Theorem~\ref{apx_thrm:trunc}}\label{appss:proofonenorm}
In order to apply Lemma~\ref{fact:mixing_time_to_fixedpoint} for proving Theorem~\ref{apx_thrm:trunc}, it remains to bound $\normp{ \CL^\beta - \CL^{\beta.r} }{1\to 1}$.
But this follows immediately from the proof for Lemma~\ref{lemthmrmhighT1}:
note that the bounds of $\mathcal{L}^{\beta}_{a,\alpha}$ and $\mathcal{L}^{\beta,r}_{a,\alpha}$ derived for the $\norm{.}_{\infty\rightarrow \infty}$ norm also hold for the $\norm{.}_{1\rightarrow 1}$-norm.
To see this, compare for instance with the bound in Eq.~\eqref{eq:part2}. We have for the transition part $T_{a,\alpha}^r(\rho)=L_{a,\alpha}^{\beta,r}\rho L_{a,\alpha}^{\beta,r\dagger}$
\begin{align}\label{eq:part3}
\begin{split}
\norm{T_{a,\alpha}^{r+1}-T_{a,\alpha}^{r}}_{1\rightarrow 1}&=\sup_X \frac{\norm{T_{a,\alpha}^{r+1}(X)-T_{a,\alpha}^{r+1}(X)}_1}{\norm{X}_1}
\\&=\sup_X \frac{\norm{(L_{a,\alpha}^{\beta,(r+1)}-L_{a,\alpha}^{\beta,r})X L_{a,\alpha}^{\beta,(r+1)\dagger}+L_{a,\alpha}^{\beta,r})X (L_{a,\alpha}^{\beta,(r+1)\dagger}-L_{a,\alpha}^{\beta,(r)\dagger})}_1}{\norm{X}_1}
\end{split}
\end{align}
However,
\begin{align}
\begin{split}
&\norm{(L_{a,\alpha}^{\beta,(r+1)}-L_{a,\alpha}^{\beta,r})X L_{a,\alpha}^{\beta,(r+1)\dagger}}_1\leq \norm{L_{a,\alpha}^{\beta,(r+1)}-L_{a,\alpha}^{\beta,r}}_\infty\norm{X}_1\norm{L_{a,\alpha}^{\beta,(r+1)\dagger}}_\infty
\end{split}
\end{align}
and thus 
\begin{align}
\norm{T_{a,\alpha}^{r+1}-T_{a,\alpha}^{r}}_{1\rightarrow 1}\leq 2 \norm{L_{a,\alpha}^{\beta,(r+1)}-L_{a,\alpha}^{\beta,r}}_\infty \norm{L_{a,\alpha}^{\beta,(r+1)\dagger}}_\infty.
\end{align}
The same argument holds for the other parts of the Lindbladian $\mathcal{L}^{\beta,r}_{a,\alpha}$.
From this observation, it immediately follows, using the same steps as in \cref{subsec:Tracedistancelemma}:
\begin{align}
    \sum_\alpha \normp{\mathcal{L}^\beta_{a,\alpha}-\mathcal{L}^{\beta,r}_{a,\alpha}}{1\rightarrow 1}\leq \sum_\alpha \sum_{r'=r}^{\infty}\normp{\mathcal{E}^{\beta,r'}_{a,\alpha}}{1\rightarrow 1}\leq \sum_\alpha 7 \frac{\gamma(\beta J)^{r+1}}{1-\gamma(\beta J)}=\mathcal{O}((\beta J)^{r/2})
\end{align}
Adding contributions from $n$ sites, we get
\begin{align}
 \norm{\mathcal{L}^{\beta} - \mathcal{L}^{\beta,r}}_{1\to 1} = \mathcal{O}(n(\beta J)^{r/2}).
\end{align}
Also, since we have proved the mixing time scales as $\log{n}$, it follows from Lemma~\ref{fact:mixing_time_to_fixedpoint} that the trace distance between the steady states of the original and truncated evolution is given by
\be
\|\rho_{\beta,r} - \rho_\beta\|_1 = O((\beta J)^{r/2}n\log n).
\ee
 This expression concludes our proof of Theorem~\ref{apx_thrm:trunc}.

\section{Error Bound on Local Observables}\label{section:Gap stability}

In this Appendix, we provide the proof for the following theorem.

\begin{thrm}[Truncation Error of local observables (\cref{thrm:trunc_local} restated).]  \label{apx_thrm:trunc_local}
Under the assumptions of Theorem 1, consider $\beta<\beta^*$ and let $O_X\in\mathsf S_X$ be supported on $X\subseteq\Lambda$. Then there exist constants $\gamma,\eta>0$, depending only on the model parameters and $\beta$, such that $\rho_{\beta,r}$ satisfies
\begin{equation}
	\left| \tr(\rho_{\beta, r} O_X) - \tr(\rho_\beta O_X) \right| \leq  \norm{O_X}c(|X|)\Bigl[e^{-\gamma r} +ne^{-\eta d_X}\Bigr]
\end{equation}
where $c(|X|)=O\bigl(\mathrm{poly}(|X|)\bigr)$, $n \equiv |\Lambda|$ is the number of qubits, and
\[
d_{X} \;=\; \min_{\,a\in X\,}\ell\bigl(a,\partial\Lambda\bigr)
\]
is the minimal graph‐distance from any $a\in X$ to the boundary $\partial\Lambda$.
\end{thrm}

\medskip
For a fixed region $X$, the distance $d_X$, which is distance between $X$ and to the complement $\Lambda^c$ of the lattice $\Lambda$ and is thus the distance to the boundaries, increases with system size, which implies that the term $|\Lambda| \nu_\eta^{-1}(d_X)$ vanishes as the system becomes large. This term can therefore be interpreted as a finite-size correction~\cite{cubitt2015stability}.
The proof exploits two key features: the rapid-mixing property of the Lindbladian $\cL^\beta$ and the quasilocal nature of the perturbation introduced by the truncation. Together, they allow us to invoke the stability theory for local, rapidly mixing dissipative systems developed in Ref.~\cite{cubitt2015stability}. We summarize these foundational results in \cref{ss:lindbladian stability} and present the detailed proof of Theorem~\ref{apx_thrm:trunc_local} in\cref{ss:proof_local_observables}.

\subsection{Stability of local and rapidly mixing Lindbladians}\label{ss:lindbladian stability}

We will make use of the following stability result, adapted from Ref.~\cite{cubitt2015stability}:

\begin{thrm}[Adapted from Theorem 7 in~\cite{cubitt2015stability}]\label{thrm:Lindblad stability}
Let $\cL$ be a uniform family of local Lindbladians with finite range interactions, a unique fixed point, and satisfying rapid mixing. Consider a perturbation of the form
\begin{align*}
\mathcal{E}^{\Lambda}=\sum_{a\in \Lambda}\sum_{r'\geq 0} \mathcal{E}_{a,r'},
\end{align*}
where:
\begin{enumerate}
\item Each $\mathcal{E}_{a,r'}$ acts on $B_a(r')$ and $\norm{\mathcal{E}_{a,r'}}_{1\rightarrow 1}\leq \epsilon e(r')$, where $\epsilon > 0$ is a small parameter characterizing the perturbation strength and $e(r')$ is a rapidly decaying function with $e(r')<1$.
\item $\mathcal{L}_{a,r'} + \mathcal{E}_{a,r'}$ is a Lindbladian itself for all $a$ and $r'\geq 0$.
\end{enumerate}
For any finite lattice $\Lambda$, define $\mathcal{L}^\Lambda=\sum_{B_a(r)\subset \Lambda}\cL_{a,r}$ the truncated version of $\cL$ on $\Lambda$
\begin{equation}
T_t = \exp(t\mathcal{L}^\Lambda)
\end{equation}
and
\begin{align}
S_t=\exp(t(\cL^{\Lambda}+\mathcal E^{\Lambda}))
\end{align}
Then, for any observable $O_X$ supported on $X \subset \Lambda$, the following bound holds:
\begin{align}
\norm{T_t^\dagger(O_X)-S_t^\dagger(O_X)}\leq c(|X|)\norm{O_X}(\epsilon+|\Lambda| \nu_\eta^{-1}(d_X)),
\end{align}
where $d_X = \inf_{a\in X, b\in \Lambda^c}\ell(a, b)$  is the distance between the set $A$ and the complement of $\Lambda$ with respect to $\mathbb{Z}^d$, $\nu^{-1}_\eta(x) = e^{-\eta x}$ 
and $\eta$ is a positive constant independent of both $\Lambda$ and $t$.  Moreover, $\nu_\eta^{-1}(d_X)\leq (1+d_X)^{-D-1}$, and $c(|X|)$ is polynomially bounded in $|X|$ and independent of $\Lambda$ and $t$.
\end{thrm}

Before proceeding, we make a few comments concerning the result above. Firstly, note that the result above is for a \textit{uniform family} of Lindbladians.
We refer to Definition 3 in Ref.~\cite{cubitt2015stability} for a precise definition, but informally one can think of a uniform family of Lindbladians as a function that generates a finite-size version $\mathcal{L}^\Lambda$ defined on $\Lambda$, starting from a Lindbladian $\mathcal{L}$ defined on $\mathbb{Z}^d$.
Such uniform families provide a rigorous framework for dealing with lattice size scaling. We will apply Theorem~\ref{thrm:Lindblad stability} to the uniform family $\mathcal{L}^{\beta,r}$, for which the finite-size version for any specific finite lattice $\Lambda$ is defined as $\mathcal{L}^{\beta,r,\Lambda}$.
Additionally, we note that Theorem 7 in Ref.~\cite{cubitt2015stability} is more general then the adaption stated above, in that it allows for non-Lindbladian perturbations, periodic boundary conditions, perturbations with boundary terms, and local Lindbladians with either exponentially decaying or quasilocal interactions (in which case the function $\nu_\eta$ will differ).
The proof in this case is completely analogous, but we focus here on the case of open boundary conditions to streamline the discussion.

However, none of these generalizations are necessary for our purposes, and so we have omitted details relevant to those cases from our presentation.

Additionally, as pointed out in Ref.~\cite{cubitt2015stability}:
\begin{enumerate}
\item For a fixed region $A$, the distance $d_X$, which is distance between $X$ and to the complement $\Lambda^c$ of the lattice $\Lambda$ and is thus the distance to the boundaries, increases with system size, which implies that the term $|\Lambda| \nu_\eta^{-1}(d_X)$ vanishes as the system becomes large. This term can therefore be interpreted as a finite-size correction.
\item The theorem allows $O_X$ to act on disconnected subsets of the lattice $\Lambda$, meaning that it also applies to correlation functions.
\end{enumerate}

\subsection{Proof of Theorem~\ref{thrm:trunc_local}}\label{ss:proof_local_observables}

We prove Theorem~\ref{thrm:trunc_local} by applying Theorem~\ref{thrm:Lindblad stability}. Specifically, as mentioned above, we consider as a starting point the uniform family of Lindbladians $\mathcal{L}^{\beta,r}$ for which $\mathcal{L}^{\beta,r,\Lambda}$ is defined via the RHS of Eq.~\eqref{eq:trunc_lindbl}. We also define $\mathcal{L}^{\beta,\Lambda}$ as the RHS of Eq.~\eqref{eq:LindbladLin}.  Now, let $\mathcal{E}_{a,\alpha}^{\beta,r}$ be as defined in \cref{appssbounding} below Eq.~\eqref{L_telescopic}. We then note that if we define
\begin{equation}
\mathcal{E}^{\beta,r}_{a} = \sum_{\alpha=1}^3\mathcal{E}_{a,\alpha}^{\beta,r}
\end{equation}
and 
\begin{equation}
\mathcal{E}^{\beta,r,\Lambda} = \sum_{a\in\Lambda}\sum_{r'=r}^\infty \mathcal{E}^{\beta,r'}_{a}
\end{equation}
then via Eqs.~\eqref{eq:LindbladLin},~\eqref{eq:trunc_lindbl} and~\eqref{eq:error:truncation} we have
\begin{equation}
\mathcal{L}^{\beta,\Lambda} = \mathcal{L}^{\beta,r,\Lambda} + \mathcal{E}^{\beta,r,\Lambda}.
\end{equation}
As $\mathcal{L}^{\beta,\Lambda}$ is itself a valid Lindbladian, the second condition in Theorem~\ref{thrm:Lindblad stability} is fulfilled.
To apply Theorem~\ref{thrm:Lindblad stability} for our purposes, we still need:
\begin{enumerate}
\item To prove that $\mathcal{L}^{\beta,r}$ is rapid mixing, and admits a unique fixed point.
\item To prove $\|\mathcal{E}^{\beta,r'}_a\|_{1\rightarrow 1} \leq \epsilon e(r')$ for rapidly decaying $e(r')$, for all $r'\geq r$.
\end{enumerate}
We note that rapid mixing of $\mathcal{L}^{\beta,r}$ is implied by the logarithmic mixing time of $\mathcal{L}^{\beta,r}$, as  proven in Theorem~\ref{thrm:mixingtime}~\cite{cubitt2015stability}, and that the existence of a unique fixed point was proven in Theorem~\ref{thrm:trunc}. For the final remaining condition, note that via the arguments of \cref{appss:proofonenorm} we have already proven that 
\begin{equation}
\|\mathcal{E}^{\beta,r'}_{a,\alpha}\|_{1\rightarrow 1} \leq 7\gamma(\beta J)^{r'}
\end{equation}
as long as $\beta J \leq 1/100$. Using this, we straightforwardly have for all $r'\geq r$ that 
\begin{align}
\|\mathcal{E}^{\beta,r'}_a\|_{1\rightarrow 1} &\leq \sum_{\alpha = 1}^3 \|\mathcal{E}^{\beta,r'}_{a,\alpha}\|_{1\rightarrow 1} \\
&\leq 21 \gamma(\beta J)^{r'} \\
&= \left[21 \gamma(\beta J)^{r}\right] \gamma(\beta J)^{(r'-r)}.
\end{align}
As such, we can apply Theorem~\ref{thrm:Lindblad stability} with $\epsilon = 21 \gamma(\beta J)^{r}$ and $e(r') = \gamma(\beta J)^{r'-r}$, which is rapidly decaying for all $r'\geq r$. More specifically, Defining $T_t = \exp{(t\mathcal{L}^{\beta,r,\Lambda})}$ and $S_t = \exp{(t\mathcal{L}^{\beta,\Lambda})}$ Theorem~\ref{thrm:Lindblad stability} yields
\begin{equation}
\norm{T_t^\dagger(O_A)-S_t^\dagger(O_A)}\leq c(A)\norm{O_A} \left(21 \gamma(\beta J)^{r}  +|\Lambda| \nu_\eta^{-1}(d_X)\right)
\end{equation}
which by taking $t\rightarrow \infty$  gives
\begin{equation}
\left| \tr(\rho_{\beta, r} O_A) - \tr(\rho_\beta O_A) \right| \leq c(A)\norm{O_A}\left(21 \gamma(\beta J)^{r}  +|\Lambda| \nu_\eta^{-1}(d_X)\right),
\end{equation}
which implies the claim of Theorem~\ref{thrm:trunc_local}.

\section{Proof of Lemma~\ref{lem:Quantumgate}}\label{app:ProofQuantumgate}

In this Appendix, we present a proof of Lemma~\ref{lem:Quantumgate}. We begin with a brief restatement of the relevant definitions and notations from the main text. 

We are interested in providing a compilation procedure for each elementary channel \(\exp(\tau \mathcal{L}_{a,\alpha}^{\beta,r})\), which represents the randomized evolution described in Eq.~\eqref{eq:channelrandomvariable}.
To do this, we define the Hermitian operator
\begin{align}\label{eq:V_operatorapp}
    O^{a,\alpha}_r := \begin{pmatrix}
    \sqrt{\tau}\,G_r^{a,\alpha} & L_r^{a,\alpha\dagger} \\
    L_r^{a,\alpha} & \sqrt{\tau}\,G_r^{a,\alpha}
    \end{pmatrix},
\end{align}
where \(L_r^{a,\alpha}\) and \(G_r^{a,\alpha}\) are the jump operators and associated coherent terms, respectively, as was previously defined in Eq.~\eqref{eq:trunc_LG}. 

We consider the quantum channel
\begin{equation}\label{eq:c_channelapp}
\mathcal C^{a,\alpha}_{r,\tau}(\rho) := \text{Tr}_{\text{anc}}\left(U^{a,\alpha}_r(\rho \otimes |0\rangle\langle 0|_{\text{anc}})U^{a,\alpha\dagger}_r\right),
\end{equation}
with \(\text{Tr}_{\text{anc}}(\cdot)\) denoting the partial trace over the ancilla qubit and the unitary \(U^{a,\alpha}_r\) describing the evolution over a time interval \(\sqrt{\tau}\) as
\begin{equation}\label{eq:U_time_evoapp}
U^{a,\alpha}_r = \exp(-iO^{a,\alpha}_r\sqrt{\tau}).
\end{equation}

The circuit for the implementation of the channel described by Eq.~\eqref{eq:c_channelapp} is given by
$$
\begin{array}{c}
\Qcircuit @C=1em @R=1.5em {
    & \lstick{|0\rangle} & \multigate{1}{\quad U^{a,\alpha}_r\quad} & \qw & \rstick{\text{Discard.}} \qw \\
    & \lstick{\rho}      & \ghost{\quad U^{a,\alpha}_r\quad}        & \qw & \qw
}
\end{array}
$$
The error of this implementation is given by Lemma~\ref{lem:Quantumgate}, which we restate here
\begin{lem}[Lemma~\ref{lem:Quantumgate} restated]\label{lem:Quantumgateappendix} The channel in Eq.~\eqref{eq:c_channelapp} satisfies
\be\label{eq:gadget_error2}
\left\|\mathcal C^{a,\alpha}_{r,\tau}-\exp(\tau \mathcal{L}_{a,\alpha}^{\beta,r})\right\|_\diamond = O(\tau^2).
\ee
\end{lem}

To prove this lemma, we first state the following auxiliary lemma, a standard result in functional analysis (e.g., see Chapter XIII, \S 6 of \cite{lang1993real}):

\begin{lem}\label{Taylor}
Let $X$ be a Banach space and let $\mathcal{B}(X)$ denote the space of bounded linear operators on $X$, equipped with an operator norm $ \| \cdot \|_{\mathcal{B}(X)}$. Suppose $ F: \mathbb{R} \to \mathcal{B}(X) $ is $ (K+1)$ -times continuously differentiable. Then, for any $ x \in \mathbb{R} $, the Taylor expansion of $ F(x)$ about  $x = 0$ satisfies:
\begin{align*}
F(x) = \sum_{m=0}^{K} \frac{x^m}{m!} F^{(m)}(0) + R_{K,F}(x),
\end{align*}
where $F^{(m)}$ denotes $m$th derivative and the remainder \( R_{K,F}(x) \) satisfies the operator norm bound:
\begin{align*}
\| R_{K,F}(x) \|_{\mathcal{B}(X)} \leq \frac{|x|^{K+1}}{(K+1)!} \sup_{s \in [0, |x|]} \| F^{(K+1)}(s) \|_{\mathcal{B}(X)}.
\end{align*}
\end{lem}

With this lemma, we prove Lemma~\ref{lem:Quantumgateappendix}.
\begin{proof}[Proof of Lemma~\ref{lem:Quantumgateappendix}]
A proof excluding a coherent term was given in Ref.~\cite{chen2024randomizedmethodsimulatinglindblad}. The present version extends that result to incorporate the modification introduced by an additional coherent term.

Let use the following notation,
\be
T_{m,F} := \frac{1}{m!} F^{(m)}(0).
\ee
Then, consider the Taylor expansion of both operators in the lemma statement as
\begin{align}
\begin{split}
    &\mathcal C^{a,\alpha}_{r,\tau}= \mathcal I + \tau T_{2, C^{a,\alpha}_{r,\tau}}+R_{3,C^{a,\alpha}_{r,\tau}}(\sqrt{\tau}),\\
    &\exp(\tau \mathcal{L}_{a,\alpha}^{\beta,r})= \mathcal I + \tau T_{1, \exp(\tau \mathcal{L}_{a,\alpha}^{\beta,r})}+R_{1,\exp(\tau \mathcal{L}_{a,\alpha}^{\beta,r})}(\tau).
\end{split}
\end{align}
where $\mathcal I$ is the identity superoperator.
Note that terms odd in \(\sqrt{\tau}\) in the first expansion are canceled due to tracing out the ancilla qubit.  
Furthermore, the lowest-order terms in this expansion coincide,
\begin{align}
    T_{1, \exp(\tau \mathcal{L}_{a,\alpha}^{\beta,r})}=T_{2, C^{a,\alpha}_{r,\tau}} = \mathcal{L}_{a,\alpha}^{\beta,r}.
\end{align}
Thus, using the triangle inequality, we obtain
\begin{align}
   \| \exp(\tau \mathcal{L}_{a,\alpha}^{\beta,r}) - C^{a,\alpha}_{r,\tau}\|_{\diamond} &= \|R_{1,\exp(\tau \mathcal{L}_{a,\alpha}^{\beta,r})}(\tau)-R_{3,C^{a,\alpha}_{r,\tau}}(\sqrt{\tau})\|_{\diamond}\nonumber \\
   &\leq \|R_{1,\exp(\tau \mathcal{L}_{a,\alpha}^{\beta,r})}(\tau)\|_{\diamond}+\|R_{3,C^{a,\alpha}_{r,\tau}}(\sqrt{\tau})\|_{\diamond}\leq C \tau^2,
\end{align}
where we used Lemma~\ref{Taylor} on both terms separately and $C$ is a constant. This completes the proof of Lemma~\ref{lem:Quantumgateappendix}.
\end{proof}

\section{Additional Models and Numerics Details}

\subsection{Additional Models}\label{sec:Further models}

We consider two additional spin-chain models. The objective is to show that our approach remains effective for integrable models and models with symmetries, indicating that it does not rely on full thermalization.

\subsubsection{1d transverse-field Ising model}
\begin{figure*}[t!]
\centering
\includegraphics[width=1.0\textwidth]{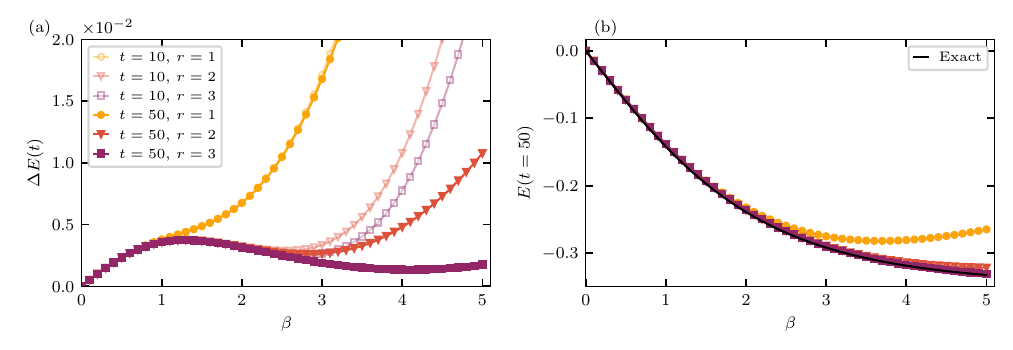}
\caption{\textbf{Energy convergence across temperatures in the transverse-field Ising model.} Energy density convergence for the transverse-field Ising Hamiltonian in Eq.~\eqref{eq:TFI} on $n=12$ sites as a function of inverse temperature $\beta$, using a Trotter step size $\tau = 0.1$. (a) Relative energy error $\Delta E(t)$ after $t = 10$ and $t = 50$: errors decrease monotonically with increasing $\beta$, while the growing gap at large $\beta$ signals mixing times $t_{\mathrm{mix}} > 50$; variations across truncation radii $r$ reflect residual truncation inaccuracies. (b) Absolute energy compared to exact diagonalization: a truncation radius of $r = 3$ achieves near–ground-state fidelity up to high $\beta$, with slight deviations at larger $\beta$ due to finite‐radius effects.
}
\label{fig:PlotTFI}
\end{figure*}

In the following, we present results for a simplified version of the model in Eq.~\eqref{eq:MFI}, namely the transverse-field Ising model with periodic boundary conditions,
\begin{align}\label{eq:TFI}
    H = \sum_{\langle i,j\rangle} S^z_i S^z_{j} + g \sum_{i=1}^n S^x_i.
\end{align}
Throughout, we set $g = 0.6$, placing the system within the ordered phase of the zero-temperature quantum phase transition at $g = 1.0$. Unlike the mixed-field Ising model, the transverse-field Ising model is integrable and can be mapped to a one-dimensional system of spinless free fermions via a Jordan–Wigner transformation~\cite{suzuki2012quantum}. However, the thermalizing Lindbladian dynamics we consider enables dynamics among the different sectors and thus breaks the integrability.

\Cref{fig:PlotTFI} shows the convergence of the energy toward the thermal value. As with the mixed-field Ising model, we observe rapid convergence to the Gibbs-state energy density using a small truncation radius, $r \leq 3$. Similarly, Fig.~\ref{fig:TFIEbeta} demonstrates that correlation functions are accurately reproduced with the same truncation. These results confirm that the algorithm remains effective even in integrable settings.

\begin{figure}[t!]
\includegraphics[width=0.5\textwidth]{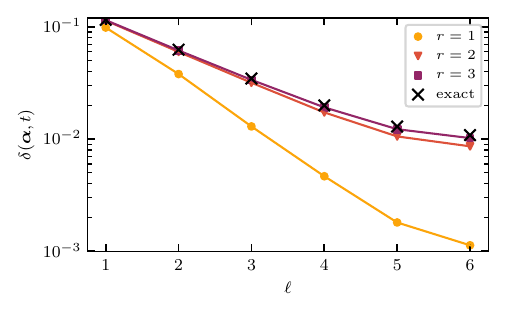}
\caption{ \textbf{Spatial profile of the two-spin correlators.} The correlator $\delta(\bm{a},t)$, for the transverse-field Ising model defined in Eq.~\eqref{eq:TFI}. The Trotter step size is $\Delta t=0.1$, $\beta=3.0$. Results of the correlator are shown after $t=10$ and truncation radius $r=1$~(yellow), $r=2$~(orange), $r=3$~(red) and the exact results~(black crosses).  For a truncation radius $r\geq2$, the proposed protocol can reproduce the correlator at large distances $\ell$.}
\label{fig:TFIEbeta}
\end{figure}
\begin{figure*}[t!]
\centering
\includegraphics[width=1.0\textwidth]{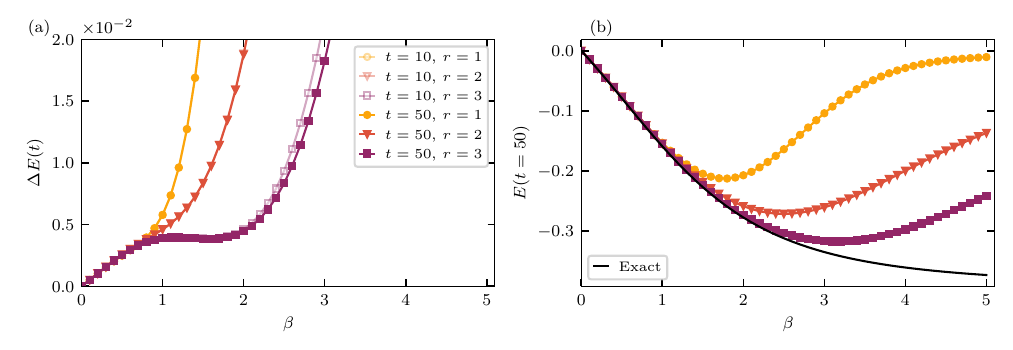}\caption{\textbf{Energy convergence across temperatures in the XXZ model.} Energy density convergence for the XXZ Hamiltonian in Eq.~\eqref{eq:XXZ} on $n=12$ sites as a function of inverse temperature $\beta$, using a Trotter step size $\tau= 0.1$. (a) Relative energy error $\Delta E(t)$ after $t = 10$ and $t = 50$ steps: the error decreases monotonically with increasing $\beta$, while the widening gap between the two curves highlights increasingly slow mixing at larger $\beta$. Even a truncation radius $r = 3$ fails to achieve rapid convergence at long times. (b) Absolute energy compared with exact diagonalization: deviations at high $\beta$ (low temperatures) reflect residual truncation errors and extended mixing times.
}
\label{fig:PlotXXZ}
\end{figure*}

\subsubsection{XXZ model}

Finally, we present results for the antiferromagnetic XXZ  model
\begin{align}\label{eq:XXZ}
    H= \sum_{\langle i,j\rangle} (S^x_iS^x_{i+1}+S^y_iS^y_{j}) + \Delta \sum_{\langle i,j\rangle} S^z_iS^z_{j}
\end{align}
with $\Delta=0.6$, which puts the model far from the isotropic point $\Delta = 1$ and in the middle of the gapless phase that exhibits quasi-long-range correlations at zero temperature. The results are shown in Fig.~\ref{fig:PlotXXZ}.
While the results indicate that our algorithm still works even in this case for high temperatures,
the results for large inverse temperatures $\beta>1$ show more significant deviations from the exact thermal state than in comparison to the previous models.

\subsection{Renormalization of the Lindbladian}\label{subsec:Normalization}
In \cref{subsec:Envelope} and the corresponding Fig.~\ref{fig:Filterfunction}, we compare the effect of different envelope functions $q(\nu)$ in Eq.~\eqref{eq:envelope}. For better comparison, we fix the jump operators for different envelope functions to have the same norm as the Gaussian case in Eq.~\eqref{eq:Gauss}.

\begin{figure}[t!]
\includegraphics[width=0.5\textwidth]{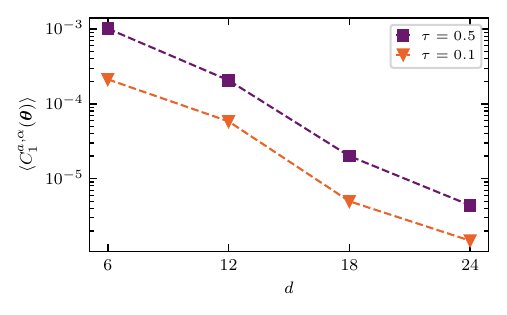}
\caption{\textbf{Variational compilation accuracy as a function of template depth.} Final loss values obtained by optimizing circuit templates of depth $d$ for different Trotter step sizes $\tau$, for fixed truncation radius $r=1$. Each curve corresponds to a fixed $\tau$, showing a systematic, monotonic decrease in compilation error as the template depth increases.}
\label{fig:Compilationgradient}
\end{figure}

Let $L^{\beta,r}_{a,\alpha}|_{q(\nu)}$ and $G^{\beta,r}_{a,\alpha}|_{q(\nu)}$ denote the jump operators and coherent terms, respectively, introduced in Eq.~\eqref{eq:trunc_LG}  evaluated using the envelope function $q(\nu)$. Then, we compute the mean of the norms for all operators:
\begin{align}
\varphi_{\beta,r}[q(\nu)]:=\frac{1}{3}\sum_{\alpha=1}^3 \left\|L^{\beta,r}_{a,\alpha}\bigl|_{q(\nu)}\right\|_2.
\end{align}
Note that the norm does not depend on the qubit index $a$ due to translational invariance. The normalized jump operators are then obtained
by
\begin{align}
L^{\beta,r}_{a,\alpha}\Bigl|_{q(\nu),\text{norm.}}\frac{\varphi_{\beta,r}\left[\exp\left(-\frac{(\beta\nu)^2}{8}\right)\right]}{\varphi_{\beta,r}[q(\nu)]}L^{\beta,r}_{a,\alpha}\Bigl|_{q(\nu)}.
\end{align}
In order to maintain detailed balance, the coherent terms should be normalized as
\begin{align}
G^{\beta,r}_{a,\alpha}|_{q(\nu),\text{norm.}}=\left(\frac{\varphi_{\beta,r}[\exp\left(-\frac{(\beta\nu)^2}{8}\right)]}{\varphi_{\beta,r}[q(\nu)]}\right)^2G^{\beta,r}_{a,\alpha}|_{q(\nu)}.
\end{align}
The performance of the renormalized jump operators is shown in Fig.~\ref{fig:Filterfunction}.

\subsection{Optimization details}
\label{optimiz_scetion}
In this subsection, we provide details of the optimization procedure of the compilation using the Adam optimizer~\cite{adam_kingma_2014}.

As the hyperparameters, we choose a learning rate $l_r=10^{-3}$, first-moment decay rate $\beta_1=0.99$, second-moment decay rate $\beta_2=0.99$, and regularization $\epsilon=10^{-3}$. For a given template, we initialize 50 different random configurations, iterate 8000 steps and choose the best configuration after the optimization process.
\Cref{fig:Compilationgradient} shows the optimal outcomes, averaged over the unitaries $U^{a,1}_r$, $U^{a,2}_r$ and $U^{a,3}_r$.

\subsection{Choice of loss function in variational compilation} \label{app:loss_function_details}
\begin{figure}[t!]
	\includegraphics[width=0.95\textwidth]{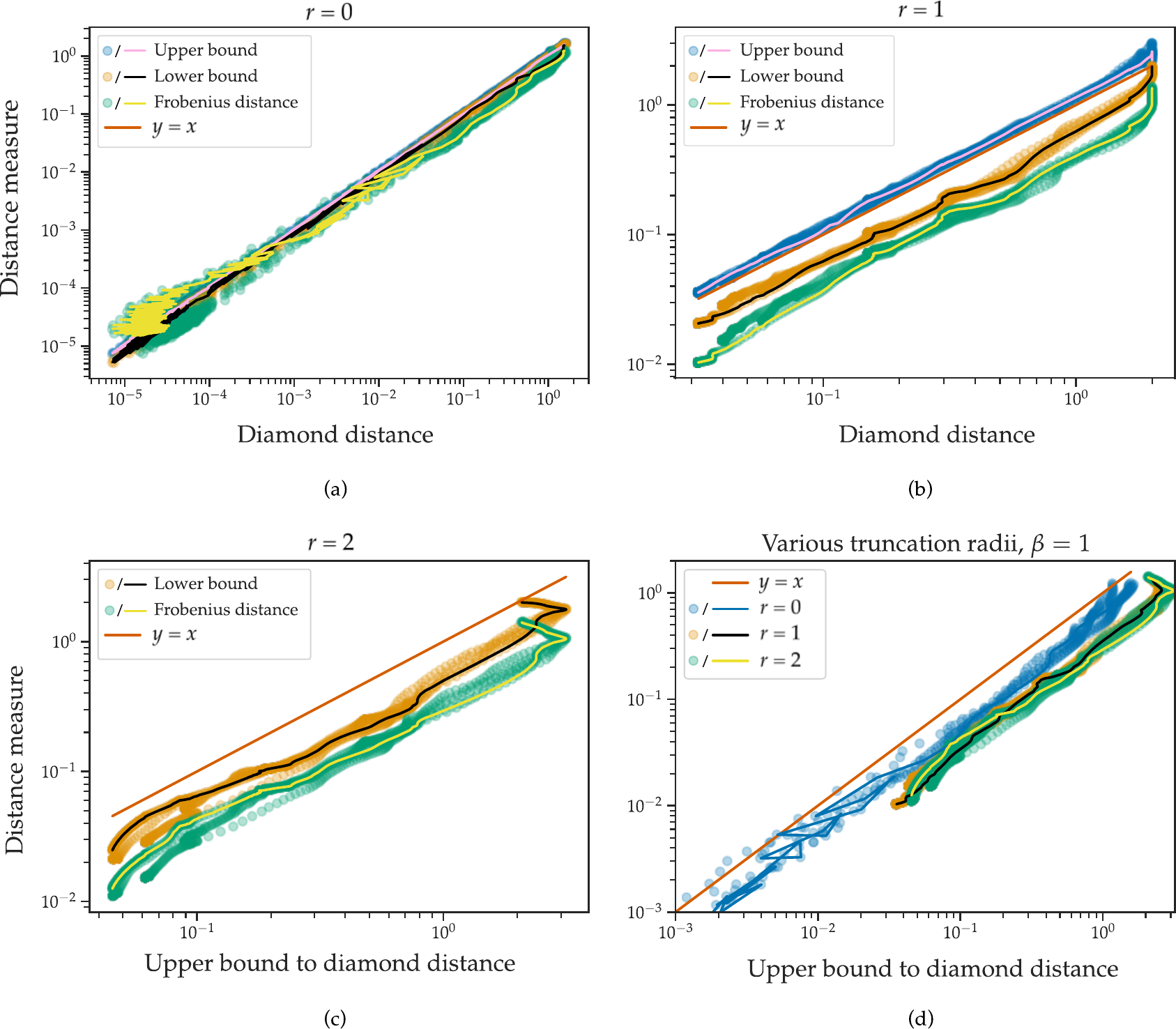}
    \caption{\textcolor{black}{\textbf{Correlation between diamond-distance proxies.} Distance measures between the target quantum channel and the trained circuit during variational compilation, shown for five independent trajectories illustrating the evolution of the distance under optimization from different random circuit initializations. Dots represent raw data points, while solid lines depict a representative example trajectory. Results are shown for different gadget constructions: (a) a one-qubit gadget ($r=0$), (b) a \textit{T}-gadget ($r=1$), and (c) a long-\textit{T}-gadget ($r=2$). (d) Same data as (a)--(c), focusing on the modified Frobenius distance and the upper bound to the diamond distance for various radii $r$ of the channel support.}}
	\label{fig:distance_correlations}
\end{figure}
In this work we consider the diamond distance between channels and two proxies for the diamond distance, namely its upper and lower bounds. This is because while computing the diamond distance of two channels from their specified Choi-Jamio\l{}kowski matrix is formally efficient via semidefinite programs \cite{watrous2009semidefiniteprogramscompletelybounded}, its upper and lower bounds \cite{nechita2018almostallquantumchannels} are even more efficiently computable and particularly amenable for use in a framework employing automatic differentiation, such as JAX.
Specifically, suppose a channel $\mathcal{C} : D_{X_1} \to D_{X_2}$ mapping input states with dimension $d_1$ to states with dimension $d_2$ has a Choi-Jamio\l{}kowski matrix given by
\begin{align}
	J(\mathcal C) = \frac{1}{d_1} \sum_{i,j=1}^{d_1} \ketbra{i}{j} \otimes \mathcal{C}\left( \ketbra{i}{j} \right). \label{eq:choijmatrix}
\end{align}
The upper and lower bounds on the diamond distance $\norm{\mathcal{C}_1-\mathcal{C}_2}_\diamond$ between two channels $\mathcal{C}_1$ and $\mathcal{C}_2$ can be written in terms of the difference in their Choi-Jamio\l{}kowski matrix\footnoteI $J(\mathcal{C}_1) - J(\mathcal{C}_2) =: J(\Delta \mathcal{C})$ \cite{nechita2018almostallquantumchannels}.
Specifically, the lower bound is 
\begin{align}
	\norm{\mathcal{C}_1-\mathcal{C}_2}_\diamond \geq \norm{J(\Delta \mathcal{C})}_1, 
\end{align}
while the upper bound is
\begin{align}
	\norm{\mathcal{C}_1-\mathcal{C}_2}_\diamond \leq 
	2^{d_1} \norm{\Tr_2{\abs{J(\Delta \mathcal{C})}}}_{\infty},
\end{align}
where $\Tr_2$ denotes tracing out the second register in \cref{eq:choijmatrix}, namely the register to which the channel is applied, and $\abs{A}$ denotes $\sqrt{A^\dag A}$.

In \cref{fig:distance_correlations}, we compare these distance measures with each other for gadgets involving $2r+1$ system qubits, with the central qubit coupled to one additional ancilla qubit.
These plots show the progress of variational compilation with random initialization for 5 different initializations each.
The ``Frobenius distance'' refers to the square root of the quantity in \cref{eq:lossfunction}.
We observe that the various distance proxies are well-correlated with each other and with the diamond distance across a large range of values.
In \cref{fig:distance_correlations}(c), we do not compute the diamond distance and instead use its upper bound as a proxy, having built evidence in \cref{fig:distance_correlations}(a)--(b) that it correlates very well with the diamond distance.

In Fig.~\ref{fig:distance_correlations}(d), we plot the (modified) Frobenius distance and the upper bound to the diamond distance from these same runs. We illustrate data from different truncation radii on the same plot to emphasize that there is no dimension-dependent loss in going between the Frobenius distance and the upper bound to the diamond distance.
\end{widetext}
\bibliography{references}
\end{document}